\documentclass[journal]{IEEEtran}
\usepackage{cite}
\usepackage{amsmath,amssymb,amsfonts}
\usepackage{algorithmic}
\usepackage{graphicx}
\usepackage{textcomp}
\usepackage{xcolor}
\usepackage{nomencl}
\usepackage{verbatim}
\usepackage{caption}
\usepackage{subcaption}
\usepackage{tabularray}
\usepackage{amsthm}
\usepackage{multirow}
\usepackage{epstopdf}
\def\BibTeX{{\rm B\kern-.05em{\sc i\kern-.025em b}\kern-.08em
    T\kern-.1667em\lower.7ex\hbox{E}\kern-.125emX}}

\usepackage{etoolbox}
\newtheoremstyle{mystl}
  {0} 
  {0} 
  {\itshape}
  {}
  {\bfseries}  
  {.} 
  {5pt plus 1pt minus 1pt} 
 {\thmname{#1} \thmnumber{#2}\ifblank{#3}{}{ (\thmnote{#3})}} 
\theoremstyle{mystl}

\newtheorem{thm}{Theorem}
\newtheorem{prop}{Proposition}
\newtheorem{asm}{Assumption}
\newtheorem{rem}{Remark}
\newtheorem{lem}{Lemma}

\begin{document}

\title{Performance-Barrier Event-Triggered\\ PDE Control of Traffic Flow 
}

\author{ Peihan Zhang$^{1}$, \IEEEmembership{Student Member, IEEE} Bhathiya Rathnayake$^{1}$, \IEEEmembership{Student Member, IEEE}, Mamadou Diagne, \IEEEmembership{Senior Member, IEEE}, and Miroslav Krstic, \IEEEmembership{Fellow, IEEE}
\thanks{$^{1}$The first two authors contributed equally to the development of this contribution. P.  Zhang, M. Diagne and M. Krstic  are with the Department of Mechanical and Aerospace Engineering, University of California San Diego, 9500 Gilman Dr, La Jolla, CA 92093. Email: \{pez004, mdiagne, mkrstic\}@ucsd.edu}
\thanks{B. Rathnayake is with the Department of Electrical and Computer Engineering, University of California San Diego, 9500 Gilman Dr, La Jolla, CA 92093. Email: brm222@ucsd.edu. Corresponding author: B. Rathnayake.} 
}

\maketitle

\begin{abstract}
For stabilizing stop-and-go oscillations in traffic flow by actuating a variable speed limit (VSL) at a downstream boundary of a freeway segment, we introduce event-triggered PDE backstepping designs employing the recent concept of performance-barrier event-triggered control (P-ETC). Our design is for linearized hyperbolic Aw-Rascle-Zhang (ARZ) PDEs governing traffic velocity and density. Compared to continuous feedback, ETC provides a piecewise-constant VSL commands---more likely to be obeyed by human drivers. Unlike the existing ``regular'' ETC (R-ETC), which enforces conservatively a strict decrease of a Lyapunov function, our performance-barrier (P-ETC) approach permits an increase, as long as the Lyapunov function remains below a performance barrier, resulting in fewer control updates than R-ETC. To relieve VSL from continuously monitoring the triggering function, we also develop periodic event-triggered  (PETC) and self-triggered  (STC) versions of both R-ETC and P-ETC. These are referred to as R/P-PETC and R/P-STC, respectively, and we show that they both guarantee Zeno-free behavior and exponential convergence in the spatial $L^2$ norm. With comparative simulations, we illustrate the benefits of the performance-barrier designs through traffic metrics (driver comfort, safety, travel time, fuel consumption). The proposed algorithms reduce discomfort nearly in half relative to driver behavior without VSL, while tripling the driver safety, measured by the average dwell time, relative to the R-ETC frequent-switching VSL schedule.
\end{abstract}





\section{Introduction}

\subsection{Boundary control of ARZ traffic model: an embodiment of  coupled  hyperbolic PDE systems}

Traffic congestion refers to the situation where the number of vehicles on the road exceeds its effective capacity.  This occurrence leads to major setbacks in economic development, primarily due to the time lost by drivers, unproductive fuel consumption, and excess of carbon dioxide (CO$_2$) emissions, among other contributing factors. In congested freeways, a common sight is the ``stop-and-go" phenomenon, where vehicles are frequently forced to come to a halt due to heavy traffic, leading to hazardous and uncomfortable driving conditions.  Various  macroscopic models have been conceived  to enhance understanding of traffic flow dynamics. The controlling  of traffic systems uses ramp metering to regulate the on-ramp flow rate by traffic light  and Varying Speed Limit (VSL) actuators.  Aw–Rascle–Zhang (ARZ)  \cite{awResurrection2000},  first-order Lighthill and Whitham and Richards (LWR)   \cite{whithamLinear2011} as well as the second-order Payne-Whitham (PW) model \cite{jModel1971} 
are useful models that have been proven to adequately serve control goals. The
ARZ model has the advantages of 1) successfully capturing the anisotropic dynamics of the traffic flow given the fact that drivers mainly react to up front traffic conditions; 2) being physically reasonable to avoid backward-propagating traffic; 3) reflecting accurately the stop-and-go-like instabilities.


Recent control-oriented results that relate to   a rich set of ARZ traffic congestion models have shown promise in enhancing traffic management \cite{yu2022traffic}. These findings  expand on an early PDE backstepping control design \cite{Vazquez2011} for  $2\times 2$ linear hyperbolic systems in the canonical setting.  The stabilization of PDE model of traffic systems has seen other  advancements, particularly through the application of Lyapunov methods. Matrix inequality and gain conditions that ensure exponential stability when employing Proportional (P) or Proportional-Integral (PI) boundary feedback control laws resulting from a Lyapunov analysis are derived in \cite{zhang2017necessary} and \cite{terrand2019adding,zhangPI2019}, respectively. Studies have addressed the challenges posed by interesting traffic scenarios, which encompass interconnected highways \cite{zhang2021boundary},  integration of Adaptive Cruise Control-equipped (ACC-equipped) vehicles \cite{bekiaris2020pde}, traffic systems featuring Connected/Automated Vehicles (CAVs) \cite{qi2022delay}, or stabilization of moving shockwaves \cite{yuBilateral2019,bastin2019exponential}. From an optimal control perspective where minimizing the total traveling time is the objective function, PDE models of traffic systems  have led to several contributions \cite{bayen2004network,gugat2005optimal}.  Finally, exponentially stabilizing controllers for nonlinear hyperbolic traffic flow systems have recently been developed in \cite{karafyllis2018feedback,karafyllis2019feedback}.

\subsection{Sampled-data and event-based control of PDE systems }
Despite the rich literature on boundary control of the ARZ traffic model, it remains evident that the assumption of drivers responding promptly to a continuously updating advisory speed is not realistic. An immediate workaround is discretizing the continuous-time control law and implementing it as sampled-data control in a zero-order hold fashion. However, while a lower sampling rate of discretization is more realistic for drivers to adhere to the advisory speed, there is a caveat: the stability or convergence properties ensured by the continuous-time control may no longer be valid. As a result, it is crucial to establish  the theoretical maximum allowable sampling interval of sampling schedules that maintain the desired closed-loop system properties. This upper limit must be determined based on worst-case scenarios, regardless of how rare or infrequent they might be. Consequently, the sampling schedules usually need to be chosen conservatively. Event-triggered control (ETC) provides a systematic solution to tackle the conservativeness of sampled-data control by bringing feedback into control update tasks.  The control input is updated only when triggered by an appropriate event triggering mechanism based on system states and is held constant between events. This approach removes the need to confine the sampling period to a worst-case value, allowing for fewer control updates while preserving a satisfactory closed-loop system performance.

In recent times, progress have been made in the domain of sampled-data control and ETC for both parabolic and hyperbolic PDE systems. For parabolic PDEs, several key contributions in sampled-data control can be highlighted by the works of \cite{fridman2012robust,karafyllis2018sampled,katz2022sampled}. In the realm of ETC for parabolic PDEs, references include \cite{espitia2021event,katz2020boundary,rathnayake2021observer,rathnayake2022sampled,rathnayake2024observer,wang2022event}. Conversely, for hyperbolic PDEs, sampled-data control has been extensively studied in papers such as \cite{davo2018stability,karafyllis2017sampled,wang2022sampled}. The area of ETC of hyperbolic systems is well-covered by works like \cite{espitiaObserverbased2020,wang2022eventb,diagne2021event,espitiaTrafficFlowControl2022,espitiaEventBasedBoundaryControl2018}. Among these results, only \cite{diagne2021event} deals with nonlinear hyperbolic systems. Our  contribution advances several  early studies in the field including  \cite{espitia2020event, espitiaTrafficFlowControl2022,espitiaObserverbased2020}. The work of \cite{espitia2020event} elucidates the design of an ETC, utilizing varying speed limits (VSL) to suppress stop-and-go traffic oscillations. The study \cite{espitiaTrafficFlowControl2022} proposes an observer-based ETC that simultaneously stabilizes the traffic flow on two connected roads. Further, \cite{espitiaObserverbased2020} focus on ETC for linear 2 $\times$ 2 hyperbolic systems, which can be regarded as a generalization of the linearized ARZ model. The studies \cite{katz2020boundary,rathnayake2021observer,rathnayake2022sampled,rathnayake2024observer,wang2022event,wang2022eventb,espitia2020event,espitiaTrafficFlowControl2022,espitiaObserverbased2020,espitiaEventBasedBoundaryControl2018}, spanning both parabolic and hyperbolic PDEs, are based on dynamic event-triggering mechanisms first introduced in the seminal work \cite{girardDynamic2015} for systems described by ordinary differential equations. 

One limitation of ETC strategies is that they require continuous monitoring of triggering functions, impeding digital implementation.   We use the term \textit{continuous-time event-triggered control (CETC)} to refer to these strategies. One solution is to check the event-triggering function periodically. This approach is commonly referred to as \textit{periodic event-triggered control (PETC)} \cite{heemelsPeriodic2013}, where the triggering function is evaluated at regular time intervals. Although the triggering function is checked periodically, the control input is still updated aperiodically, coinciding with events. An alternative solution is \textit{self-triggered control (STC)} \cite{heemelsIntroduction2012}, which predicts the next event time at the current event time, thereby eliminating the need for continuous monitoring of event-triggering functions. Both PETC and STC maintain the resource efficiency of CETC, as control updates are made aperiodically and exclusively at event times. Additionally, these strategies are amenable to digital implementations. In the past few years, quite interesting studies have been conducted on both PETC \cite{heemelsPeriodic2013, wangPeriodic2020} and STC \cite{heemelsIntroduction2012, yiDynamic2019} for ODE systems. To the best of our knowledge, studies devoted to PETC and STC strategies for infinite-dimensional systems include \cite{rathnayake2023observer,rathnayake2023observerhjk,rathnayake2023prfmnce, wakaikiStability2022, wakaikiEventTriggered2020}. However, none of these studies address coupled hyperbolic PDEs like the ARZ model. 

\subsection{Results}
Leveraging the recently introduced \textit{performance-barrier based ETC (P-ETC)} for nonlinear ODEs \cite{ong2023performance} and its adaptation to boundary control of a class of parabolic PDEs \cite{rathnayake2023prfmnce}, our work applies P-ETC to variable speed limit (VSL) boundary control of the linearized inhomogeneous ARZ model. This  approach results in substantially longer intervals between events (\textit{dwell-times}) compared to the dynamic ETC strategies \cite{espitia2020event,espitiaEventBasedBoundaryControl2018}  previously applied to the linearized ARZ model.

The triggering mechanisms discussed in \cite{espitiaEventBasedBoundaryControl2018} and \cite{espitia2020event} enforce a monotonic decrease in the closed-loop system's Lyapunov function. We classify these strategies at a broader level as \textit{regular} ETC (R-ETC), distinguishing them from the P-ETC introduced in the present contribution. The monotonic decrease of the Lyapunov function is achieved by ensuring its time derivative remains strictly negative. This approach certifies that the Lyapunov function decreases faster than a specific exponentially decaying signal, which depends on initial data and is known as the \textit{performance-barrier}. Drawing on previous research \cite{ong2023performance} and \cite{rathnayake2023prfmnce}, allowing deviations from a monotonically decreasing Lyapunov function, while still adhering to the performance barrier, might prolong the duration between events. To enable this leeway in the Lyapunov function's behavior, we incorporate the so-called \textit{performance residual} into the event-triggering mechanism. This residual is defined as the difference between the performance barrier and the Lyapunov function. Consequently, by design, the P-ETC allows for longer dwell-times in any given state, compared to the R-ETC. Notably, this is achieved without inducing Zeno behavior in the closed-loop system, while still maintaining adherence to the performance barrier, leading to the exponential convergence of the system states to zero in the spatial $L^2$ norm. Since the triggering function requires continuous monitoring in order to detect events, we refer to this strategy specifically as P-CETC. For similar reasons, we refer to the strategies in \cite{espitiaEventBasedBoundaryControl2018} and \cite{espitia2020event} as R-CETC.

Building upon the techniques introduced in \cite{rathnayake2023observer} for parabolic PDEs, we further aim to circumvent the need for continuous monitoring of the triggering functions in the R-CETC and P-CETC. To achieve this, we extend these methods to PETC and STC, resulting in what we refer to as R- and P-PETC and R- and P-STC, respectively. Both R-PETC and P-PETC employ periodic event-triggering functions, which are established by deriving explicit upper bounds on the underlying continuous-time event-triggering functions. Since the triggering functions are evaluated periodically, Zeno behavior is inherently absent in both R-PETC and P-PETC. In the case of STC, we develop the R- and P-STC by designing state dependent functions with uniform and positive lower bounds, which when evaluated at the current control update time produces the waiting time until the next control update. These functions are designed via obtaining upper bounds on the variables that constitute the R- and P-CETC event-triggering functions. Both R-STC and P- STC are also inherently Zeno-free, since they maintain a uniform, positive lower bound for dwell-times. The R-PETC and R-STC force the Lyapunov function to strictly decrease along the closed-loop system solution, while the P-PETC and P-STC  permit occasional increases in the Lyapunov function, as long as it stays below the established performance barrier, resulting in longer dwell-times between events compared to their regular counterparts. All the introduced PETC and STC strategies guarantee the exponential convergence of the closed-loop system states to zero in the spatial $L^2$ norm.

The sparsity of P-ETC VSL updates, owing to longer dwell-times compared to R-ETC updates can lead to improved driving safety. Over a certain period, when the VSL is updated less frequently, fewer drivers are distracted by changes in the VSL while passing through the VSL zone after having already adjusted their speed once. Distracting drivers from tasks critical for safe driving to focus on a competing activity may result in insufficient or no attention being paid to essential driving activities, as noted in \cite{sheridan2004driver}. Safe driving involves maintaining a safe distance from the vehicle ahead, and in this context, a competing activity would involve reacting to the changing VSL. If the VSL changes frequently, drivers might either disregard the VSL suggestions, leading to stop-and-go oscillations and hence, an uncomfortable driving experience among other negative outcomes, or compromise their safety by focusing on frequent speed adjustments instead of maintaining safe distances between vehicles. However, P-ETC achieves three times longer average dwell times than R-ETC, thereby providing VSL schedules that drivers can adhere to without compromising safety, while also reducing discomfort nearly in half compared to driver behavior without VSL.    

\subsection{Contributions}
\noindent\textit{Major contributions:}
\begin{itemize}
\item Design of P-ETC for boundary control of the linearized ARZ model, leading to sparser control updates compared to the class of ETC strategies \cite{espitiaEventBasedBoundaryControl2018} and \cite{espitia2020event} applied to the linearized ARZ model. 
\item The first PETC and STC approach for coupled linear hyperbolic PDEs, specifically extending P-CETC to P-PETC and P-STC to avoid continuous monitoring of the P-CETC event-triggering function required for event detection. None of the prior works \cite{rathnayake2023observer,rathnayake2023observerhjk,rathnayake2023prfmnce, wakaikiStability2022, wakaikiEventTriggered2020} have dealt with PETC and STC of coupled PDEs.
\item Demonstration that P-ETC enables a trade-off between the level of safety and driver comfort in traffic management through the tuning of a parameter $c \geq 0$, referred to as the \textit{resource-aware parameter}.
\end{itemize}

\noindent \textit{Other contribution:}
\begin{itemize}
\item Extension of the R-CETC \cite{espitiaEventBasedBoundaryControl2018} and \cite{espitia2020event} to R-PETC and R-STC to avoid continuous monitoring of the R-CETC event-triggering function.
\end{itemize}

The interrelation among various technical and principal results in the paper are depicted in Fig. \ref{results_all}.

\begin{figure}
    \centering
\includegraphics[width=1\linewidth]{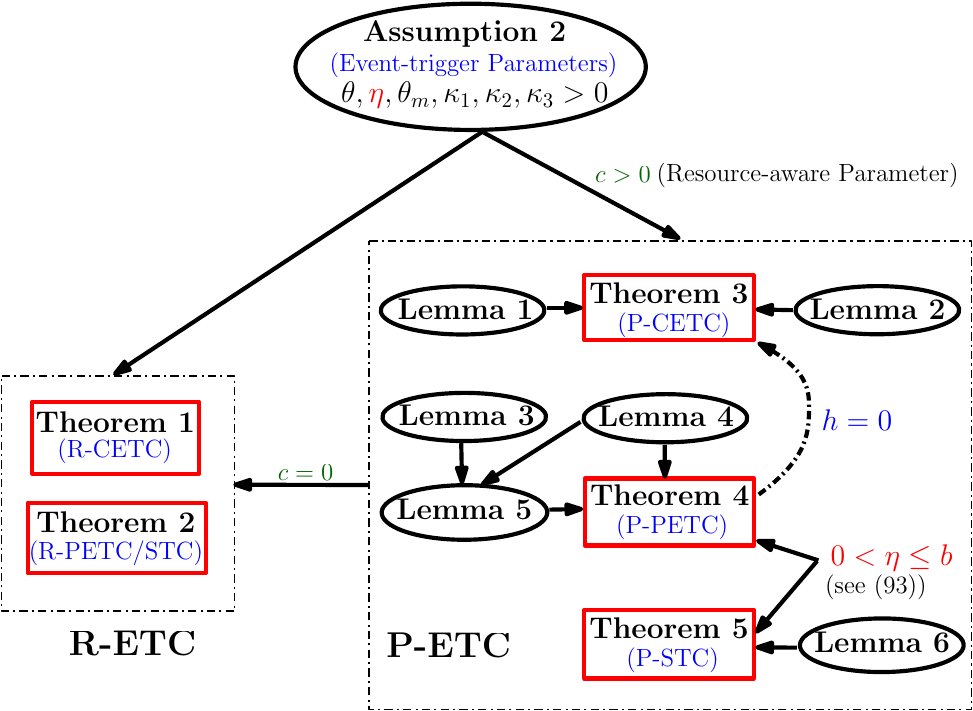}
    \caption{Interrelation among various technical and principal
results in the paper.}
    \label{results_all}
\end{figure}

\subsection{Notation}
$\mathbb{R}_{+}$ is the positive real line while $\mathbb{N}$ is the set of natural numbers.
Let $\alpha:[0,\ell] \times \mathbb{R}_{+} \rightarrow \mathbb{R}$ be given. 
$\alpha[t]$ denotes the profile of $\alpha$ at certain $t \geq 0$, i.e., $(\alpha[t])(x)=\alpha(x, t)$, for all $x \in[0,\ell]$. 
The set of all functions $g:[0,\ell] \rightarrow \mathbb{R}^n$ such that $\int_0^\ell g(x)^Tg(x) d x<\infty$ is denoted by $L^2\left([0,\ell], \mathbb{R}^n\right)$.
Given a topological set $S$, and an interval $I \subseteq \mathbb{R}$, the set $\mathcal{C}^0(I ; S)$ is the set of continuous functions $g: I \rightarrow S$.
Variables and functions related to \textit{regular} ETCs are denoted with a superscript $r$ while those of \textit{performance-barrier} ETCs are denoted with $p$.

\subsection{Organization}
The rest of the paper is organized as follows: Section \ref{sec:ctn} introduces the inhomogeneous ARZ model alongside its continuous-time control and emulation. Section \ref{sec:regular} details the regular event-triggered control (R-ETC), consisting of preliminary R-CETC, and the newly developed R-PETC and R-STC. In Section \ref{sec:perf}, the performance-barrier event-triggered control (P-ETC), including P-CETC, P-PETC, and P-STC, is discussed. Simulations are presented in Section \ref{sec:sim}, followed by conclusions in Section \ref{sec:Conclusions}.

\section{Continuous-time Control and Emulation}
\label{sec:ctn}
In this section, we briefly present the ARZ model under continuous-time PDE backstepping control. This is followed by its emulation for ETC.

\subsection{Aw–Rascle–Zhang (ARZ) model}
The inhomogeneous ARZ model is a second-order nonlinear hyperbolic PDE system that describes the relationship between traffic density $\rho(x,t)$ and velocity $v(x,t)$ as given by 
\begin{align}
    \partial_t \rho+\partial_x(\rho v) & =0, \label{eq:ARZ-rho}\\
    \partial_t v+\left(v-\rho p^{\prime}(\rho)\right) \partial_x v & =\frac{V(\rho)-v}{\tau} \label{eq:ARZ-v}.
\end{align}
Here, the term $p(\rho)$ is the traffic pressure, an increasing function of density $\rho(x,t)$ given by
\begin{equation}
    p(\rho)=c_0\rho^\gamma,
\end{equation}
where $c_0, \gamma \in \mathbb{R}_{+}$. The term $\tau$ is the relaxation time related to the time scale of drivers' behavior adapting to the equilibrium density-velocity profile. The term $V(\rho)$ describes the velocity-density relationship at the equilibrium $(\rho^\star,v^\star)$, as given by Greenshield’s model \cite{greenshieldsSTUDY1935}
\begin{equation}
    V(\rho)=v_f\left(1-\Big(\frac{\rho}{\rho_m}\Big)^{\gamma}\right),
\end{equation}
where $v_f$ is the maximum velocity, $\rho_m$ is the maximum density and $v^{\star}=V\left(\rho^{\star}\right)$.

\begin{asm}[Boundary conditions of the ARZ model]
    We assume a constant traffic flux $q^{\star}=\rho^{\star} v^{\star}$ entering the domain from $x=0$, while a varying speed limit (VSL) is imposed at the outlet $x=\ell$, where $\ell>0$ is the road length. Therefore, the boundary conditions are
    \begin{align}
            \rho(0, t) & =\frac{q^{\star}}{v(0, t)}, \label{eq:bd-cond-ARZ-rho} \\
            v(\ell, t) & =U(t)+v^{\star} \label{eq:bd-cond-ARZ-v},
    \end{align}
    where $U(t)$ represents the variation from the steady-state velocity $v^{\star}$ and will be designed later. Additionally, we assume that drivers adhere to $v(\ell, t)$ as indicated on the VSL signs.
\end{asm}

The linearized ARZ model around the steady state $\left(\rho^{\star}, v^{\star}\right)$ with boundary conditions is given by
\begin{align}
    \partial_t \tilde{\rho}+v^{\star} \partial_x \tilde{\rho} & =-\rho^{\star} \partial_x \tilde{v}, \label{eq:linear-ARZ-rho}\\
    \partial_t \tilde{v}-\left(\rho^{\star} p^{\prime}\left(\rho^{\star}\right)-v^{\star}\right) \partial_x \tilde{v} & =\frac{\tilde{\rho} V^{\prime}\left(\rho^{\star}\right)-\tilde{v}}{\tau}  \label{eq:linear-ARZ-v}, \\
    \tilde{\rho}(0, t) & = -\frac{\rho^{\star}}{v^{\star}} \tilde{v}(0, t), \label{eq:bd-cond-linear-ARZ-rho}\\
    \tilde{v}(\ell, t) & = U(t) \label{eq:bd-cond-linear-ARZ-v},
\end{align}
where $(\tilde{\rho}(x, t), \tilde{v}(x, t))$ are the deviations from the equilibrium and are defined as $\tilde{\rho}(x, t)=\rho(x, t)-\rho^{\star}$, $\tilde{v}(x, t)=v(x, t)-v^{\star} $.
By following the transformations (see \cite{yuVarying2018}) 
\begin{align}
\bar{w}(x, t) & =\exp \left(\frac{c_1}{v^{\star}} x\right)\left(\frac{\gamma p^{\star}}{\rho^{\star}} \tilde{\rho}(x, t)+\tilde{v}(x, t)\right), \\
\bar{v}(x, t) & =\exp \left(\frac{c_2}{\gamma p^{\star}-v^{\star}} x\right) \tilde{v}(x, t),
\end{align}
system \eqref{eq:linear-ARZ-rho}-\eqref{eq:bd-cond-linear-ARZ-v} can be mapped to a first-order 2 $\times$ 2 hyperbolic system in $(\bar{w}, \bar{v})$ as
\begin{align}
\partial_t \bar{w}+v^{\star} \partial_x \bar{w} & =\bar{c}_1(x) \bar{v} , \label{eq:sys-wv-wbar}\\
\partial_t \bar{v}-\left(\gamma p^{\star}-v^{\star}\right) \partial_x \bar{v} & =\bar{c}_2(x) \bar{w} , \label{eq:sys-wv-vbar}\\
\bar{w}(0, t) & =-r_0 \bar{v}(0, t), \label{eq:sys-wv-bd-cond-w}\\
\bar{v}(\ell, t) & =r_1 U(t) \label{eq:sys-wv-bd-cond-v},
\end{align}
where
\begin{align}\label{cxx}
 \bar{c}_1(x)&=\exp \left(\frac{c_1}{v^{\star}} x-\frac{c_2}{\gamma p^{\star}-v^{\star}} x\right) c_2, \\
 \bar{c}_2(x)&=-\exp \left(\frac{c_2}{\gamma p^{\star}-v^{\star}} x-\frac{c_1}{v^{\star}} x\right) c_1,
\end{align}
with \begin{align}\label{para1}r_0  &=\frac{\gamma p^{\star}-v^{\star}}{v^{\star}},\quad  r_1 = \exp \left(\frac{c_2}{\gamma p^{\star}-v^{\star}} \ell\right),\\\label{para2}c_1 &=\frac{1}{\tau} \frac{v_f}{\rho_m} \frac{\rho^{\star}}{\gamma p^{\star}}, \quad  c_2 =\frac{1}{\tau}\left(\frac{v_f}{\rho_m} \frac{\rho^{\star}}{\gamma p^{\star}}-1\right).\end{align} The parameters $c_1$, $c_2$, and $r_0$ satisfy
    \begin{equation}\label{cndtns}
            c_1>\frac{1}{\tau}>0, \quad
            c_2=c_1-\frac{1}{\tau}>0, \quad
            r_0>0 ,
    \end{equation}
which represents the instability condition of  \eqref{eq:linear-ARZ-rho}-\eqref{eq:bd-cond-linear-ARZ-v} within congested regime \cite{yuVarying2018} . 
Our objective is to achieve exponential convergence of $(\bar{w},\bar{v})$ to zero in the spatial $L^2$ norm. 

\subsection{Continuous-time PDE Backstepping Control}

Consider the invertible backstepping transformation
\begin{align}
\alpha(x, t)=&\bar{w}(x, t)-\int_0^x K^{11}(x, \xi) \bar{w}(\xi, t) d \xi \nonumber \\
&-\int_0^x K^{12}(x, \xi) \bar{v}(\xi, t) d \xi , \label{eq:K1}\\
\beta(x, t)= & \bar{v}(x, t)-\int_0^x K^{21}(x, \xi) \bar{w}(\xi, t) d \xi \nonumber \\
& -\int_0^x K^{22}(x, \xi) \bar{v}(\xi, t) d \xi , \label{eq:K2}
\end{align}
where $K^{i j}(x, \xi), i, j=1,2$ are the kernels that evolve in the triangular domain $\mathcal{T}=\{(x, \xi): 0 \leq \xi \leq x \leq \ell\}$ and are governed by equations detailed in \cite{yuVarying2018}.
Then, adopting standard arguments in PDE backstepping, we can show that the transformation \eqref{eq:K1},\eqref{eq:K2}, and the continuous-time boundary control law $U(t)$ derived in \cite{yuVarying2018}, given by
\begin{equation}
\label{eq:U-wbar-vbar}
    \begin{aligned}
    U(t)
    \!=\!  \frac{1}{r_1} \!\int_0^\ell  \!\! \left(  K^{21} (\ell, \xi) \bar{w}(\xi, t) 
    \!+\!   K^{22} (\ell, \xi) \bar{v}(\xi, t) \right) d \xi,
    \end{aligned}
\end{equation}
map the system \eqref{eq:sys-wv-wbar}-\eqref{cndtns} into the target $(\alpha, \beta)$-system:
\begin{align}
\partial_t \alpha+v^{\star} \partial_x \alpha & =0, \label{eq:sys-ab-a} \\
\partial_t \beta-\left(\gamma p^{\star}-v^{\star}\right) \partial_x \beta & =0, \label{eq:sys-ab-b}\\
\alpha(0, t) & =-r_0 \beta(0, t), \label{eq:sys-ab-bd-cond-a}\\
\beta(\ell, t) & =0 \label{eq:sys-ab-bd-cond-b}.
\end{align}

The inverse transformation of \eqref{eq:K1},\eqref{eq:K2} is given by:
\begin{align}
\bar{w}(x, t)=\alpha(x, t)+ & \int_0^x L^{11}(x, \xi) \alpha(\xi, t) d \xi \nonumber\\
& +\int_0^x L^{12}(x, \xi) \beta(\xi, t) d \xi, \label{eq:L1} \\
\bar{v}(x, t)=\beta(x, t)+ & \int_0^x L^{21}(x, \xi) \alpha(\xi, t) d \xi \nonumber \\
& +\int_0^x L^{22}(x, \xi) \beta(\xi, t) d \xi, \label{eq:L2}
\end{align}
where kernels $L^{i j}(x, \xi), i, j=1,2$ are a specific case of the general form of kernel equations detailed in \cite{Vazquez2011}. The input $U(t)$ can also be expressed in target system $(\alpha,\beta)$ states as
\begin{equation}
\label{eq:U-alpha-beta}
\begin{aligned}
    U(t) \!=\! \frac{1}{r_1} \! \int_0^\ell \!\! \left( L^{21}(\ell, \xi) \alpha(\xi,t) + L^{22}(\ell, \xi) \beta(\xi,t) \right) d \xi.
\end{aligned}
\end{equation}

\subsection{Emulation of the PDE Backstepping Control}
We aim to achieve exponential convergence of the states of the system \eqref{eq:sys-wv-wbar}-\eqref{cndtns} to zero by sampling the continuous-time controller $U (t)$ given by \eqref{eq:U-wbar-vbar} at a sequence of time instants $\{t_k\}_{k\in \mathbb N}$.
These time instants will be determined via several event triggers in subsequent sections.
The control input is held constant between two successive time instants and is updated when a certain condition is met.
We define the control input for $t \in\left[t_k, t_{k+1}\right), k \in \mathbb{N}$ as
\begin{equation}
\label{eq:Uk}
\begin{aligned}
    &U_k := U(t_k) \\& = \frac{1}{r_1}\int_0^\ell K^{21}(\ell, \xi) \bar{w}(\xi, t_k) d \xi +\frac{1}{r_1} \int_0^\ell K^{22}(\ell, \xi) \bar{v}(\xi,t_k) d \xi\\&= \frac{1}{r_1}\int_0^\ell L^{21}(\ell, \xi) \alpha(\xi,t_k) d \xi +\frac{1}{r_1} \int_0^\ell L^{22}(\ell, \xi) \beta(\xi,t_k) d \xi.
\end{aligned}
\end{equation}
As a result, the boundary conditions in \eqref{eq:bd-cond-ARZ-v},\eqref{eq:bd-cond-linear-ARZ-v},\eqref{eq:sys-wv-bd-cond-v} become $v(\ell, t) =U_k+v^{\star}$, $\tilde v(\ell, t) =U_k$, and
\begin{equation}
\label{eq:sys-wv-bd-cond-Uk}
    \bar{v}(\ell, t) =r_1 U_k.
\end{equation}
The actuation deviation $d(t)$ between the continuous-time control and its sampled counterpart, i.e., the input holding error, is defined as follows for $t \in\left[t_k, t_{k+1}\right), k \in \mathbb{N}$:
\begin{equation}
\begin{aligned}
\label{eq:d(t)}
d(t) := & U_k-U(t) \\
= &\frac{1}{r_1} \int_0^\ell L^{21}(\ell, \xi)\left(\alpha\left(\xi, t_k\right)-\alpha(\xi, t)\right) d \xi \\
& +\frac{1}{r_1} \int_0^\ell L^{22}(\ell, \xi)\left(\beta\left(\xi, t_k\right)-\beta(\xi, t)\right) d \xi .
\end{aligned}
\end{equation}
Note that we have expressed $d(t)$ in terms of the the target system states $(\alpha,\beta)$. Through backstepping transformation \eqref{eq:K1},\eqref{eq:K2}, the system \eqref{eq:sys-wv-wbar}-\eqref{eq:sys-wv-bd-cond-w},\eqref{cxx}-\eqref{cndtns},\eqref{eq:Uk},\eqref{eq:sys-wv-bd-cond-Uk} is mapped to the target system
\begin{align}
\alpha_t(x, t)+v^{\star} \alpha_x(x, t) & =0, \label{eq:sys-ab-d-a}\\
\beta_t(x, t)-\left(\gamma p^{\star}-v^{\star}\right) \beta_x(x, t) & =0, \label{eq:sys-ab-d-b}\\
\alpha(0, t) & =-r_0 \beta(0, t), \label{eq:sys-ab-d-bd-cond-a}\\
\beta(\ell, t) & =r_1 d(t) \label{eq:sys-ab-d-bd-cond-b},
\end{align}
for $t \in\left[t_k, t_{k+1}\right), k \in \mathbb{N}$.

Now we present the well-posedness of the closed-loop system \eqref{eq:sys-wv-wbar}-\eqref{eq:sys-wv-bd-cond-w},\eqref{cxx}-\eqref{cndtns},\eqref{eq:Uk},\eqref{eq:sys-wv-bd-cond-Uk} between two sampling instants.

\begin{prop}[Well-Posedness between control updates]
\label{prop:wellpose}
For given $\left(\bar w (\cdot, t_k), \bar v (\cdot, t_k)\right)^T \in L^2\left((0,\ell); \mathbb{R}^2\right)$, there exists a unique solution $(\bar w, \bar v)^T \in \mathcal{C}^0\left(\left[t_k, t_{k+1}\right] ; L^2\left((0,\ell) ; \mathbb{R}^2\right)\right)$  to the system \eqref{eq:sys-wv-wbar}-\eqref{eq:sys-wv-bd-cond-w},\eqref{cxx}-\eqref{cndtns},\eqref{eq:Uk},\eqref{eq:sys-wv-bd-cond-Uk}, between two time instants $t_k$ and $t_{k+1}$.
\end{prop}

\begin{rem} \rm
\label{rem:wellpose}
    This proposition is a straightforward application of Proposition 1 in \cite{espitiaObserverbased2020}, with the difference up to the scaling factor $\ell$, i.e., the road length.
    Throughout the paper, we establish the well-posedness of the closed-loop system by iteratively constructing the solution in the hybrid time domain $\mathbb{T}=\bigcup_{k=0}^{K-1}\left[t_k, t_{k+1}\right] \times\{k\}$ where $\mathbb{T} \subset \mathbb{R}_{\geq 0} \times \mathbb{N}$ and $K$ is possibly $\infty$ and/or $t_K=\infty$.
\end{rem}

\section{Regular Event-Triggered Control (R-ETC)}
\label{sec:regular}

This section presents the designs for regular ETC (R-ETC). These designs are associated with strictly decreasing Lyapunov functions, and are presented in three configurations: continuous-time event-triggered (R-CETC), periodic event-triggered (R-PETC), and self-triggered control (R-STC).

\subsection{Regular Continuous-time Event-triggered Control (R-CETC)}

The design of R-CETC is detailed in \cite{espitia2020event,espitiaEventBasedBoundaryControl2018}. However, the parameters in the triggering mechanism and the conditions for parameter selection are presented in a complex way. In this paper, we present the details of a more straightforward approach to the triggering mechanism and parameter choices. Since R-CETC serves as a foundation for R-PETC and R-STC, below we summarize the main results of R-CETC.

Let $I^r=$ $\left\{t_0^r, t_1^r, t_2^r, \ldots\right\}$ denote the sequence of event-times associated with R-CETC, and it consists of two parts:
\begin{enumerate}
    \item An ETC input $U_k^r$
    \begin{equation}
        \label{eq:U-R-CETC}
            \begin{aligned}
                U_k^r := U(t_k^r),
            \end{aligned}
    \end{equation}
    for $t \in\left[t^r_k, t^r_{k+1}\right), k \in \mathbb{N},$ where $U(t)$ is given by \eqref{eq:U-wbar-vbar}. Then, the boundary condition \eqref{eq:sys-wv-bd-cond-Uk} becomes
    \begin{equation}
    \label{eq:R-CETC-bd}
        \bar{v}(\ell, t) =r_1 U_k^r.
    \end{equation}
    \item A continuous-time event-trigger determining event-times,
    \begin{equation}
    \label{eq:R-CETC-trigger-t}
        t^r_{k+1}=\inf \left\{t \in \mathbb{R}_{+} \mid t>t_k^r, \Gamma^r(t)>0 , k \in \mathbb{N} \right\},
    \end{equation}
    with $t_0^r=0$, where $\Gamma^r(t)$  is the triggering function defined as
        \begin{equation}
    \label{eq:R-CETC-trigger-func}
        \Gamma^r(t) :=  d^2(t) - \theta m^r(t).
    \end{equation}
The function $d(t)$ is given by \eqref{eq:d(t)} for $t \in \left[t^r_k, t^r_{k+1}\right), k \in \mathbb{N}$, and $m^r(t)$ satisfies the ODE
    \begin{equation}
    \label{eq:ODE-dot-m-r}
        \begin{aligned}
            \dot{m}^r(t)  =& -\eta m^r(t)- \theta_m d^2(t)+\kappa_1\Vert\alpha[t]\Vert^2+\kappa_2\Vert\beta[t]\Vert^2\\&+\kappa_3{\alpha}^2(\ell, t),
        \end{aligned}
    \end{equation}
    for $t \in \left(t^r_k, t^r_{k+1}\right), k \in \mathbb{N}$ with $m^r\left(t^r_0\right)=m^r(0)>0$ and $m^r\left(t_k^{r-}\right)=m^r\left(t_k^r\right)=m^r\left(t_k^{r+}\right)$. The parameters  $\theta, \eta,\theta_m,\kappa_1,\kappa_2,\kappa_3>0$ are event-trigger parameters to be appropriately chosen.  
\end{enumerate}
Below, we outline the conditions on event-trigger parameters that ensure the Zeno-free behavior and the exponential convergence of the closed-loop 
signals of the system \eqref{eq:sys-wv-wbar}-\eqref{eq:sys-wv-bd-cond-w},\eqref{cxx}-\eqref{cndtns},\eqref{eq:U-R-CETC}-\eqref{eq:ODE-dot-m-r} to zero in the spatial $L^2$ norm. 

\begin{asm}[Event-trigger parameter selection]
\label{asm:R-CETC-param}
The parameters $\theta,\eta>0$ are arbitrary design parameters, and $\kappa_1,\kappa_2,\kappa_3>0$ are chosen as
\begin{equation}\label{betas}
\kappa_{1}=\frac{\varepsilon_{1}}{\theta(1-\sigma)},\hspace{5pt}\kappa_{2}=\frac{\varepsilon_{2}}{\theta(1-\sigma)},\hspace{5pt}\kappa_{3}=\frac{\varepsilon_{3}}{\theta(1-\sigma)},
\end{equation}
where $\sigma\in(0,1)$ and 
\begin{align}
\begin{split}
\label{al1}\varepsilon_{1}&= 4\frac{(v^{\star})^2}{r_1^2}\int_0^\ell( \dot{L}^{21}(\ell,y))^2dy,
\end{split}\\
\label{al2}
\varepsilon_{2}&=4\frac{\left(\gamma p^{\star}-v^{\star}\right)^2}{r_1^2}\int_0^\ell( \dot{L}^{22}(\ell,y))^2dy,\\\label{gg3}
\varepsilon_{3}&=4\frac{(v^\star)^2}{r_1^2}(L^{21}(\ell,\ell))^2.
\end{align} The event-trigger parameter $\theta_m>0$ is chosen as 
\begin{equation}\label{vvbnml}
\theta_m = Cr_1^2r_0^2e^{\frac{\mu\ell}{\left(\gamma p^{\star}-v^{\star}\right)}},
\end{equation}
where $\mu>0$,
\begin{align}\label{CC}
    C>&\max\Bigg\{e^{\frac{\mu\ell}{v^\star}}\kappa_3, \frac{\max\{\kappa_1,\kappa_2\}r}{\mu}\Bigg\},
\end{align}
with $r$ defined as
\begin{equation}\label{rr}
    r := \frac{1}{\min\Big\{\frac{1}{v^\star}e^{-\frac{\mu\ell}{v^\star}},\frac{r_0^2}{\left(\gamma p^{\star}-v^{\star}\right)}\Big\}}.
\end{equation}
\end{asm}

Next we summarize the main results under R-CETC.

\begin{thm}
[Results under R-CETC]
\label{thm:R-CETC}
Let $I^r= \left\{t_0^r, t_1^r, t_2^r, \ldots\right\}$ with $t_0^r = 0$ be the set of event-times generated by the R-CETC approach \eqref{eq:U-R-CETC}-\eqref{eq:ODE-dot-m-r} with appropriate choices for the event-trigger parameters under Assumption \ref{asm:R-CETC-param}. 
Then, it holds that
    \begin{align}
        \Gamma^r(t) \leq 0, \forall t \in [0, \sup \left(I^r\right) ).
    \end{align}
As a result, the following results hold:
\begin{enumerate}
     \item[R1:]  The set of event-times $I^r$ generates an increasing sequence. It holds that $t_{k+1}^r-t_k^r \geq \tau_d> 0, k \in \mathbb{N}$, where
    \begin{align}
    \label{eq:tau_d}
        {\tau_d} =\frac{1}{a} \ln \left(1+\frac{\sigma a}{(1-\sigma)(a+\theta\theta_m)}\right). 
    \end{align}
for $\sigma\in (0,1)$. Here,  $a>0$ is given by
    \begin{align}
        a & = 1+\varepsilon_0+\eta, \label{eq:a-def}
    \end{align}
    where
    \begin{align}
        \varepsilon_0 &= 4{(\gamma p^{\star}-v^{\star})^2} \big(L^{22}(\ell, \ell)\big)^2\label{eq:eps-2}.
    \end{align}
    Due to the existence of a uniform positive minimal dwell-time ${\tau_d}>0$, it follows that $t_k^r \rightarrow \infty$ as $k \rightarrow \infty$, thereby guaranteeing Zeno-free behavior.

    \item[R2:] For every $\left(\bar w (\cdot, 0), \bar v (\cdot, 0)\right)^T \in L^2\left((0,\ell); \mathbb{R}^2\right)$, there exists a unique solution $(\bar w, \bar v)^T \in \mathcal{C}^0\left(\mathbb{R}_+ ; L^2\left((0,\ell) ; \mathbb{R}^2\right)\right)$ to the system \eqref{eq:sys-wv-wbar}-\eqref{eq:sys-wv-bd-cond-w},\eqref{cxx}-\eqref{cndtns},\eqref{eq:U-R-CETC}-\eqref{eq:ODE-dot-m-r} for all $t>0$.
    
    \item[R3:] The dynamic variable $m^r(t)$ governed by \eqref{eq:ODE-dot-m-r} with $m^r(0)>0$ satisfies $m^r(t)>0$ for all $t > 0$.

    \item[R4:] Consider a Lyapunov candidate 
    \begin{align}
        \label{eq:Lyap-V-r}
        V^r(t) = V_1(t)+m^r(t),
    \end{align}
    where
\begin{equation}
    \label{eq:Lyap-V}
    \begin{aligned}
    & V_1(t):=\int_0^\ell\left(\frac{C}
    {v^{\star}} \alpha^2(x, t) e^{-\frac{\mu x}{v^{\star}}}\right. \\
    &\left.\qquad\qquad+\frac{Cr_0^2}{\gamma p^{\star}-v^{\star}} \beta^2(x, t) e^{\frac{\mu x}{\gamma p^{\star}-v^{\star}}}\right) d x,
    \end{aligned}
    \end{equation}
Then, it holds that
\begin{align}\label{Vr_dot1}
    \dot{V}^r(t) \leq -b^{\star }V^r(t),
    \end{align}
for all $t\in (t_k^r,t_{k+1}^r),k\in\mathbb{N}$, and 
\begin{align}
    \label{eq:Lyap-V-r-bd}
        V^r(t) \leq e^{-b^\star t} V_0,
    \end{align}
for all $t>0$, where $V_0 = V^r(0)$ and
    \begin{align}
    \label{eq:b*}
        b^\star :=  \min\big\{b, \eta\big\}>0,
    \end{align}
    with 
    \begin{align}\label{bbcfgj}
        b:=\mu-\frac{\max\{\kappa_1,\kappa_2\}r}{C}>0.
    \end{align}
See Assumption \ref{asm:R-CETC-param} for details on $\mu,\kappa_1,\kappa_2,r,C,\eta>0$. 
    \item[R5:] The closed-loop signal $\Vert\bar{w}[t]\Vert+\Vert\bar{v}[t]\Vert$ associated with the system \eqref{eq:sys-wv-wbar}-\eqref{eq:sys-wv-bd-cond-w},\eqref{cxx}-\eqref{cndtns},\eqref{eq:U-R-CETC}-\eqref{eq:ODE-dot-m-r}, exponentially converges to zero.
\end{enumerate}
\end{thm}

See the Appendix for the proof.

\begin{rem}\rm
We refer to the signal \( e^{-b^\star t} V_0 \) in \eqref{eq:Lyap-V-r-bd} as the \textit{performance barrier}, which the Lyapunov function of the system must not violate. The estimate for the time derivative of \( V^r(t) \) provided in \eqref{Vr_dot1} indicates that R-CETC enforces a strict decrease in the Lyapunov function \( V^r(t) \) in \eqref{eq:Lyap-V-r} along system trajectories. However, this strict requirement limits R-CETC’s ability to achieve sparser control updates. Addressing this limitation is the objective of our design for performance-barrier event triggers, detailed in Section \ref{sec:perf}.
\end{rem}

\subsection{Regular Periodic Event-triggered Control (R-PETC)}
In this subsection, we the introduce R-PETC approach applied to the system \eqref{eq:sys-wv-wbar}-\eqref{cndtns}. Since the R-PETC design can be derived via the P-PETC approach detailed in Section \ref{subsec:P-PETC}, we will only present its structure here to prevent redundancy. 

Let $\tilde I^r=$ $\left\{\tilde t_0^r, \tilde t_1^r, \tilde t_2^r, \ldots\right\}$ denote the sequence of event-times associated with R-PETC, and let the event-trigger parameters $\theta, \eta,\theta_m,\kappa_1,\kappa_2,\kappa_3>0$ be selected as outlined in Assumption \ref{asm:R-CETC-param}.
The proposed R-PETC strategy consists of two parts:
\begin{enumerate}
    \item An ETC input $\tilde U_k^r$
    \begin{equation}
        \label{eq:U-R-PETC}
            \begin{aligned}
                \tilde U_k^r := U(\tilde t_k^r),
            \end{aligned}
    \end{equation}
    for $t \in\left[\tilde t^r_k, \tilde t^r_{k+1}\right), k \in \mathbb{N},$ where $U(t)$ is given by \eqref{eq:U-wbar-vbar}. Then, the boundary condition \eqref{eq:sys-wv-bd-cond-Uk} becomes
    \begin{equation}
    \label{eq:R-PETC-bd}
        \bar{v}(\ell, t) =r_1 \tilde U_k^r.
    \end{equation}
    \item A periodic event-trigger determining event-times
    \begin{equation}
    \label{eq:R-PETC-trigger-t}
    \begin{aligned}
        \tilde t^r_{k+1}=\inf \{t \in \mathbb{R}_{+} \mid  & t>\tilde t_k^r, \tilde \Gamma^r(t)>0 , t=nh, \\
        & h>0, n\in \mathbb{N},  k \in \mathbb{N} \},
    \end{aligned}
    \end{equation}
    with $\tilde t_0^r=0$. Here $h$ is the sampling period selected as
    \begin{align}
    \label{eq:h-def}
        0<h \leq \tau_d,
    \end{align}
    where $\tau_d$ is given by \eqref{eq:tau_d}-\eqref{eq:eps-2}, and $\tilde \Gamma^r(t)$ is the triggering function defined as
    \begin{equation}
    \label{eq:R-PETC-trigger-func}
    \tilde{\Gamma}^r(t):=(a+\theta\theta_m)e^{ah}d^{2}(t )-\theta\theta_m d^2(t )-\theta a m^p(t),
    \end{equation}
    where $a$ is defined in \eqref{eq:a-def}. Further, $d(t)$ is defined in \eqref{eq:d(t)}, and $m^r(t)$ satisfies the ODE given by \eqref{eq:ODE-dot-m-r},  along the solution of  \eqref{eq:sys-wv-wbar}-\eqref{eq:sys-wv-bd-cond-w},\eqref{cxx}-\eqref{cndtns},\eqref{eq:U-R-PETC}-\eqref{eq:R-PETC-trigger-func} for all $t \in\left[\tilde t^r_k, \tilde t^r_{k+1}\right), k \in \mathbb{N}$.
\end{enumerate}

The main difference between the continuous-time event-trigger \eqref{eq:R-CETC-trigger-t}-\eqref{eq:ODE-dot-m-r} and the periodic event-trigger \eqref{eq:R-PETC-trigger-t}-\eqref{eq:R-PETC-trigger-func} lies in that the triggering function $\Gamma^r(t)$ of R-CETC has to be monitored continuously while the triggering function $\tilde \Gamma^r(t)$ of R-PETC requires only periodic evaluations.

For brevity, we present the results under R-PETC alongside R-STC results in Theorem \ref{thm:R-PETC+R-STC} in the following subsection.

\subsection{Regular Self-triggered Control (R-STC)}
\label{subsec:R-STC}
In this subsection, we present the R-STC strategy for the system \eqref{eq:sys-wv-wbar}-\eqref{cndtns}. Since this method can be derived via the P-STC design in Section \ref{subsec:P-STC}, we will only present its structure here to prevent redundancy.

Let $\check I^r=$ $\left\{\check t_0^r, \check t_1^r, \check t_2^r, \ldots\right\}$ denote the sequence of event-times associated with R-STC.
Let the event-trigger parameters $\theta, \eta,\theta_m,\kappa_1,\kappa_2,\kappa_3>0$ be selected as outlined in Assumption \ref{asm:R-CETC-param}. The proposed R-STC strategy consists of two parts:
\begin{enumerate}
    \item An ETC input $\check U_k^r$
    \begin{equation}
        \label{eq:U-R-STC}
            \begin{aligned}
                \check U_k^r := U(\check t_k^r),
            \end{aligned}
    \end{equation}
    for $t \in\left[\check t^r_k, \check t^r_{k+1}\right), k \in \mathbb{N}$ where $U(t)$ is given by \eqref{eq:U-wbar-vbar}. Then, the boundary condition \eqref{eq:sys-wv-bd-cond-Uk} becomes
    \begin{equation}
    \label{eq:R-STC-bd}
        \bar{v}(\ell, t) =r_1 \check U_k^r.
    \end{equation}
    \item A self-trigger determining event-times
    \begin{equation}
    \label{eq:R-STC-trigger-t} \check{t}_{k+1}^r=\check{t}_k^r+G^r\left(H(\check{t}_k^r), m^r\left(\check{t}_k^r\right)\right),
    \end{equation}
    with $\check t_0^r=0$ and $G^r(\cdot, \cdot)>0$ is a positively and uniformly lower-bounded function
    \begin{equation}
        \begin{aligned}
            \label{eq:R-STC-trigger-func}
            & G^r\left(H(t), m^r\left(t\right)\right) \\
            & :=\max\Bigg\{\tau_d,\frac{1}{\varrho^\star+\eta}\ln\bigg(\frac{\theta m^r(t)+\frac{\theta\theta_m H(t)}{\varrho^\star+\eta}}{H(t)+\frac{\theta\theta_mH(t)}{\varrho^\star+\eta}}\bigg)\Bigg\}.
        \end{aligned}
    \end{equation}
    Here, $\tau_d$ is R-CETC minimum dwell-time given by \eqref{eq:tau_d}-\eqref{eq:eps-2}, $m^r(t)$ satisfies the dynamics \eqref{eq:ODE-dot-m-r} along the solution of \eqref{eq:d(t)}-\eqref{eq:sys-ab-d-bd-cond-b} for $t \in\left(\check t^r_k, \check t^r_{k+1}\right), k \in \mathbb{N}$.
    The constant $\varrho^\star>0$ is defined as
\begin{align}\label{dxxxml}
    \varrho^\star := r_0^2r_1^2 e^{\frac{\mu\ell}{(\gamma p^{\star}-v^{\star})}}\varrho,
\end{align}
where 
\begin{align}\label{zzzmlw2e2}
    \varrho = \frac{4}{r_1^2}\max\Big\{v^{\star}\tilde {L}^{2 1}e^{\frac{\mu \ell}{v^{\star}}},\frac{(\gamma p^{\star}-v^{\star})\tilde {L}^{2 2}}{r_0^2}\Big\},
\end{align}
with 
 \begin{equation}\label{eq:tilde-L-21-L22}
        \tilde {L}^{2 1} \! = \!\!\! \int_0^\ell \!\! (L^{21}(\ell, \xi))^2 d \xi, \quad
        \tilde {L}^{2 2} \! = \!\!\! \int_0^\ell \!\! (L^{22}(\ell, \xi))^2 d \xi.        
\end{equation}
In \eqref{eq:R-STC-trigger-func}, $H(t)$ is defined as
\begin{align}\label{zmlsbnj}
\begin{split}
    H(t) := &3\varrho \int_{0}^{\ell}\Big(\frac{1}{v^{\star}}\alpha^2(x,t)e^{-\frac{\mu x}{v^{\star}}}\\&+\frac{r_0^2}{(\gamma p^{\star}-v^{\star})}\beta^2(x,t)e^{\frac{\mu x}{(\gamma p^{\star}-v^{\star})}}\Big)dx.
\end{split}
\end{align}
\end{enumerate}

R-STC distinguishes itself from both R-CETC and R-PETC because R-STC proactively determines the subsequent event-time based on the system state at the current event time, without monitoring any triggering function. We now state the results under R-STC and R-PETC, in the following theorem.

\begin{thm}[Results under R-STC (resp. R-PETC)]
\label{thm:R-PETC+R-STC}
    Let  $\check I^r=$ $\left\{\check t_0^r, \check t_1^r, \check t_2^r, \ldots\right\}$ with $\check t_0^r = 0$ (resp. $\tilde I^r=$ $\left\{\tilde t_0^r, \tilde t_1^r, \tilde t_2^r, \ldots\right\}$ with $\tilde t_0^r = 0$) be the set of increasing event-times generated by the R-STC approach \eqref{eq:U-R-STC}-\eqref{zmlsbnj} (resp. R-PETC approach \eqref{eq:U-R-PETC}-\eqref{eq:R-PETC-trigger-func}) with appropriate choices for the event-trigger parameters under Assumption \ref{asm:R-CETC-param}.
    Then, the following results hold:
    \begin{enumerate}
    \item[R1:] For every $\left(\bar w (\cdot, 0), \bar v (\cdot, 0)\right)^T \in L^2\left((0,\ell); \mathbb{R}^2\right)$, there exists a unique solution $(\bar w, \bar v)^T \in \mathcal{C}^0\left(\mathbb{R}_+ ; L^2\left((0,\ell) ; \mathbb{R}^2\right)\right)$ to the R-STC closed-loop system \eqref{eq:sys-wv-wbar}-\eqref{eq:sys-wv-bd-cond-w},\eqref{cxx}-\eqref{cndtns},\eqref{eq:U-R-STC}-\eqref{zmlsbnj} (resp. R-PETC closed-loop system \eqref{eq:sys-wv-wbar}-\eqref{eq:sys-wv-bd-cond-w},\eqref{cxx}-\eqref{cndtns},\eqref{eq:U-R-PETC}-\eqref{eq:R-PETC-trigger-func}) for all $t>0$.

    \item[R2:] $\Gamma^r(t)$ given by \eqref{eq:R-CETC-trigger-func} satisfies $\Gamma^r(t) \leq 0$ for all $t>0$, along the R-STC (resp. R-PETC) closed-loop solution. 
    \item[R3:] The dynamic variable $m^r(t)$ governed by \eqref{eq:ODE-dot-m-r} with $m^r(0)>0$ satisfies $m^r(t)>0$ for all $t > 0$, along the R-STC (resp. R-PETC) closed-loop solution.
    \item[R4:] The Lyapunov candidate $V^r(t)$ given by \eqref{eq:Lyap-V-r},\eqref{eq:Lyap-V} satisfies \eqref{Vr_dot1},\eqref{eq:b*},\eqref{bbcfgj} for all $t\in (\check{t}_k^r,\check{t}_{k+1}^r),k\in\mathbb{N}$ (respect. $t\in (\tilde{t}_k^r,\tilde{t}_{k+1}^r),k\in\mathbb{N}$) and \eqref{eq:Lyap-V-r-bd}-\eqref{bbcfgj} for all $t>0$, along the R-STC (resp. R-PETC) closed-loop solution.

    \item[R5:] The closed-loop signal $\Vert\bar{w}[t]\Vert+\Vert\bar{v}[t]\Vert$ associated with the R-STC (resp. R-PETC) closed-loop system exponentially converges to zero.
\end{enumerate}
\end{thm}
The proof of Theorem \ref{thm:R-PETC+R-STC} follows arguments similar to those in the proofs of Theorem \ref{thm:P-PETC} in Section \ref{subsec:P-PETC} and Theorem \ref{thm:P-STC} in Section \ref{subsec:P-STC}, and is therefore omitted.

\section{Performance-barrier Event-triggered Control (P-ETC)}
\label{sec:perf}
In this section, we discuss the design of performance-barrier ETC (P-ETC) under continuous-time event-triggered (P-CETC), periodic event-triggered (P-PETC), and self-triggered (P-STC) control. By introducing a \textit{performance residual}--the difference between the performance barrier and the Lyapunov function--into the triggering mechanism, we allow the Lyapunov function to deviate from a monotonic decrease while adhering to the performance barrier.

\subsection{Performance-barrier Continuous-time Event-triggered Control (P-CETC)}
Let $I^p=\left\{t_0^p, t_1^p, t_2^p, \ldots\right\}$ denote the sequence of event-times associated with P-CETC. Let the event-trigger parameters be selected as in Assumption \ref{asm:R-CETC-param}, and $c>0$ be an additional design parameter. The proposed P-CETC strategy consists of two parts:
\begin{enumerate}
    \item An ETC input $U_k^p$
    \begin{equation}
        \label{eq:U-P-CETC}
            \begin{aligned}
                U_k^p := U(t_k^p),
            \end{aligned}
    \end{equation}
    for $t \in\left[t^p_k, t^p_{k+1}\right), k \in \mathbb{N},$ where $U(t)$ is given by \eqref{eq:U-wbar-vbar}. Then, the boundary condition \eqref{eq:sys-wv-bd-cond-Uk} becomes
    \begin{equation}
    \label{eq:P-CETC-bd}
        \bar{v}(\ell, t) =r_1 U_k^p.
    \end{equation}
    \item A continuous-time event-trigger determining event-times
    \begin{equation}
    \label{eq:P-CETC-trigger-t}
        t^p_{k+1}=\inf \left\{t \in \mathbb{R}_{+} \mid t>t_k^p, \Gamma^p(t)>0 , k \in \mathbb{N} \right\},
    \end{equation}
    with $t_0^p=0$. The triggering function $\Gamma^p(t)$ is defined as
    \begin{equation}
    \label{eq:P-CETC-trigger-func}
    \begin{aligned}
        \Gamma^p(t) := 
         d^2(t) - \theta m^p(t)-\frac{c}{\theta_m}W^p(t),
    \end{aligned} 
    \end{equation}
    where $d(t)$ is given by \eqref{eq:d(t)}, $m^p(t)$ satisfies the ODE
    \begin{equation}
    \label{eq:ODE-dot-m-p}
        \begin{aligned}
            \dot{m}^p(t)= & -\eta m^p(t)- \theta_md^2(t)+\kappa_1\Vert\alpha[t]\Vert^2+\kappa_2\Vert\beta[t]\Vert^2\\&+\kappa_3{\alpha}^2(\ell, t)+cW^p(t),
        \end{aligned}
    \end{equation}
    for $t \in\left(t^p_k, t^p_{k+1}\right), k \in \mathbb{N}$ with $m^p\left(t^p_0\right)=m^p(0)>0,$ and $m^p\left(t_k^{p-}\right)=m^p\left(t_k^p\right)=m^p\left(t_k^{p+}\right)$. Here, $W^p(t)$ is the \textit{performance residual}, defined as the
difference between the value of the performance-barrier $e^{-b^\star t} V_0^p$ and the Lyapunov function
        \begin{equation}
        \label{eq:W-p}
            W^p(t):=e^{-b^\star t} V_0^p-V^p(t),
        \end{equation}
    where $b^\star$ is given by \eqref{eq:b*}, and $V^p(t)$ is the Lyapunov candidate defined as
    \begin{equation}
    \label{eq:Lyap-V-p}
        V^p(t) = V_1(t)+m^p(t),
    \end{equation}
    with $V_1(t)$ given by \eqref{eq:Lyap-V} and
    \begin{equation}
    \label{eq:Lyap-V0}
        V_0=V_0^p=V^p(0)=V^r(0).
    \end{equation}

\end{enumerate}

Next we present Lemma \ref{lem:P-CETC-m} and Lemma \ref{lem:P-CETC-t}, required for proving the main results of P-CETC in Theorem \ref{thm:P-CETC}.

\begin{lem}
\label{lem:P-CETC-m} Under the P-CETC event-trigger \eqref{eq:P-CETC-trigger-t}-\eqref{eq:Lyap-V0}, it holds that the triggering function $\Gamma^p(t)$ satisfies $\Gamma^p(t)\leq 0$, and as a result, the dynamic variable $m^p(t)$ governed by \eqref{eq:ODE-dot-m-p} with $m^p(0)=m^r(0)>0$ satisfies $m^p(t)>0$ for all $t\in [0, \sup \left(I^p\right) )$.
\end{lem}

The proof follows steps similar to those in Lemma 1 of \cite{rathnayake2023prfmnce} and is therefore omitted.

\begin{lem}
\label{lem:P-CETC-t}
    Assume that an event has occurred at $t=t^\star \geq 0$ under P-CETC \eqref{eq:U-P-CETC}-\eqref{eq:Lyap-V0}. If the next event time $t=t^p$ generated by P-CETC is finite, then the next event time $t=t^r$ generated by R-CETC \eqref{eq:U-R-CETC}-\eqref{eq:ODE-dot-m-r} satisfies $t^r \leq t^p$, provided that $m^r\left(t^\star\right)=m^p\left(t^\star\right)>0$ and $e^{-b^\star t} V_0 \geq V^p(t)$ for all $t \in\left[t^\star, t^r\right]$. The equality holds if $e^{-b^\star t} V_0=V^p(t)$ for all $t \in\left[t^\star, t^r=t^p\right]$.
\end{lem}

The proof follows steps similar to those in Lemma 2 of \cite{rathnayake2023prfmnce} and is therefore omitted.

Now we state the main results of P-CETC below.

\begin{thm}[Results under P-CETC]
    \label{thm:P-CETC}
    Let $I^p= \left\{t_0^p, t_1^p, t_2^p, \ldots\right\}$ with $t_0^p = 0$ be the set of event-times generated by the P-CETC approach \eqref{eq:U-P-CETC}-\eqref{eq:Lyap-V0} with appropriate choices for the event-trigger parameters under Assumption \ref{asm:R-CETC-param}, and $c>0$. 
    Then, it holds that
    \begin{align}
        \Gamma^p(t) \leq 0, \forall t \in [0, \sup \left(I^p\right) ).
    \end{align}
    As a result, the following results hold:
    \begin{enumerate}
 \item[R1:] The set of event-times $I^p$ generates an increasing sequence.
        It holds that $t_{k+1}^p-t_k^p \geq \tau_d> 0, k \in \mathbb{N}$, where the minimal dwell-time $\tau_d$ is given by \eqref{eq:tau_d}-\eqref{eq:eps-2}. 
        Since ${\tau_d}>0$, as $k \rightarrow \infty$, it follows that $t_k^p \rightarrow \infty$, thereby guaranteeing Zeno-free behavior.

        \item[R2:] For every $\left(\bar w (\cdot, 0), \bar v (\cdot, 0)\right)^T \in L^2\left((0,\ell); \mathbb{R}^2\right)$, there exists a unique solution $(\bar w, \bar v)^T \in \mathcal{C}^0\left(\mathbb{R}_+ ; L^2\left((0,\ell) ; \mathbb{R}^2\right)\right)$ to the system \eqref{eq:sys-wv-wbar}-\eqref{eq:sys-wv-bd-cond-w},\eqref{cxx}-\eqref{cndtns},\eqref{eq:U-P-CETC}-\eqref{eq:Lyap-V0} for all $t>0$.
     
        \item[R3:] The dynamic variable $m^p(t)$ governed by \eqref{eq:ODE-dot-m-p} with $m^p(0)>0$ satisfies $m^p(t)>0$ for all $t > 0$.

    \item[R4:] The Lyapunov candidate $V^p(t)$ given by \eqref{eq:Lyap-V-p},\eqref{eq:Lyap-V0} satisfies 
\begin{equation}\label{zzzbnmkfb}
    \dot{V}^p(t) \leq -b^{\star} V^p(t)+c\big(e^{-b^{\star}t}V_0-V^p(t)\big),
\end{equation}
for all $t\in (t_k^p,t_{k+1}^{p}),j\in\mathbb{N}$, and
    \begin{equation}
    \label{eq:Lyap-V-p<ebt}
        V^p(t) \leq e^{-b^\star t} V_0,
    \end{equation}
    for all $t>0$, where $b^\star$ is given by \eqref{eq:b*}.

    \item[R5:] The closed-loop signal $\Vert\bar{w}[t]\Vert+\Vert\bar{v}[t]\Vert$ associated with the system \eqref{eq:sys-wv-wbar}-\eqref{eq:sys-wv-bd-cond-w},\eqref{cxx}-\eqref{cndtns},\eqref{eq:U-P-CETC}-\eqref{eq:Lyap-V0}, exponentially converges to zero.
    \end{enumerate}
\end{thm}

\begin{proof}[\rm \textbf{Proof}]


    As a result of Lemma \ref{lem:P-CETC-m}, it holds that $\Gamma^p(t) \leq 0$, and consequently, $m^p(t)>0$ for $t \in\left[0, \sup \left(I^p\right)\right)$.
    Consider the time period $t \in\left[0, \sup \left(I^p\right)\right)$. By selecting the event-trigger parameters as outlined in Assumption \ref{asm:R-CETC-param} and $c>0$, we can show that $V^p(t)$ satisfies
        \begin{equation}
        \label{eq:Lyap-V-p-dot}
            \dot{V}^p(t) \leq-b^\star V^p(t)+c W^p(t),
        \end{equation}
        for $t \in\left(t_k^p, t_{k+1}^p\right)$ (the proof follows similar arguments to those of Theorem \ref{thm:R-CETC} provided in the Appendix).
        The time derivative of $W^p(t)$ satisfies
        \begin{equation}
            \begin{aligned}
                \dot{W}^p(t)
                 =-b^\star e^{-b^\star t} V_0^p-\dot{V}^p(t) 
                 \geq-\left(b^\star+c\right) W^p(t),
            \end{aligned}
        \end{equation}
        for $t \in\left(t_k^p, t_{k+1}^p\right)$. Then, noting that $W^p(t)$ is continuous and $W^p(0)=0$, we can obtain that
\begin{equation}
        \label{eq:W-p-proof}
        \begin{aligned}
        W^p(t) & \geq e^{-\left(b^\star+c\right)\left(t-t_k^p\right)} W^p\left(t_k^p\right) \\
        & \geq e^{\left.-\left(b^\star+c\right)\right)\left(t-t_k^p\right)} \times \prod_{i=1}^{i=k} e^{-\left(b^\star+c\right)\left(t_i^p-t_{i-1}^p\right)} W^p(0) \\
        & \geq e^{-(b^\star+c)t}W^p(0)=0,
        \end{aligned}
\end{equation}
        for all $t \in\left[0, \sup \left(I^p\right)\right)$, i.e., $e^{-b^\star t } V_0^p \geq V^p(t)$ for all $t \in\left[0, \sup \left(I^p\right)\right)$.
    This result satisfies the assumption made in Lemma \ref{lem:P-CETC-t}, from which we obtain R1.
    
    R2 is obtained recalling Proposition \ref{prop:wellpose} and Remark \ref{rem:wellpose}.
    Since R1 establishes that the system is Zeno-free, i.e., $\sup \left(I^p\right)=\infty$, we have $m^p(t)>0$ for all $t>0$ as stated in R3, $V^p(t) \leq e^{-b^\star t} V_0$ for all $t>0$ as stated in R4. Finally, we can obtain the exponential convergence of the closed-loop signals of \eqref{eq:sys-wv-wbar}-\eqref{eq:sys-wv-bd-cond-w},\eqref{cxx}-\eqref{cndtns},\eqref{eq:U-P-CETC}-\eqref{eq:Lyap-V0} to zero as stated in R5 by following classical arguments involving the bounded invertibility of the backstepping transformations \eqref{eq:K1},\eqref{eq:K2},\eqref{eq:L1},\eqref{eq:L2}.
\end{proof}

\begin{rem}\rm
   As observed from \eqref{zzzbnmkfb}, we have 
\[
\dot{V}^p(t) \leq -b^\star V^p(t) + cW^p(t),
\]
where \(W^p(t) := e^{-b^\star t} V_0^p - V^p(t)\). This indicates that the time derivative of the Lyapunov function does not necessarily need to be negative at all times. If the performance residual \(W^p(t)\) is large---meaning that the Lyapunov function \(V^p(t)\) is significantly below the performance barrier \(e^{-b^\star t} V_0^p\)---there is a reduced risk of violating the performance barrier. In such cases, \(\dot{V}^p(t)\) can be positive, allowing \(V^p(t)\) to increase. Conversely, if \(W^p(t)\) is relatively small---meaning the Lyapunov function \(V^p(t)\) is approaching the performance barrier \(e^{-b^\star t} V_0^p\)---the risk of breaching the performance barrier increases. In such situations, \(\dot{V}^p(t)\) must be negative to ensure that the Lyapunov function remains below the threshold.
\end{rem}

\begin{rem}\rm
    According to the property R4 of Theorem \ref{thm:P-CETC}, in the absence of external disturbances, the performance barrier associated with the Lyapunov function remains intact and cannot be breached. However, when disturbances are present, this barrier may be exceeded. Specifically, as time progresses to larger values, the performance residual \( W^p(t) = e^{-b^\star t} V_0 - V^p(t) \) becomes increasingly sensitive to disturbances. To enhance the robustness of the P-CETC framework against such disturbances, the triggering function can be modified as follows:
\begin{equation}
    \Gamma^p(t) := d^2(t)-\theta m^p(t)-\frac{c}{\theta_m} \max\big\{0,W^p(t)\big\},
\end{equation}
where \( m^p(t) \) satisfies the differential equation:
\begin{equation}
\begin{split}
    \dot{m}^p(t) = &  -\eta m^p(t)- \theta_md^2(t)+\kappa_1\Vert\alpha[t]\Vert^2+\kappa_2\Vert\beta[t]\Vert^2 \\
    & +\kappa_3{\alpha}^2(\ell, t) + c \max\big\{0, W^p(t)\big\},
\end{split}
\end{equation}
for \( t \in (t_k^p, t_{k+1}^p), j \in \mathbb{N} \), with the initial and continuity conditions \( m^p(t^p_0) = m^r(t^r_0) > 0 \) and \( m^p(t_k^{p-}) = m^p(t_k^p) = m^p(t_k^{p+}) \).

This adjustment introduces a safeguard mechanism: if a disturbance causes the Lyapunov function to exceed the performance barrier (i.e., \( W^p(t) < 0 \)), control transitions from P-CETC to R-CETC. Rather than being seen as a limitation, this transition is an essential safety feature of the design. During disturbances, infrequent control updates can be detrimental, and a higher frequency of updates is preferred to reduce the duration for which the plant operates in open-loop conditions. The proposed modification ensures this by transferring control to R-CETC whenever disturbances are substantial enough to breach the performance barrier. Since R-CETC enforces a strict decrease in the Lyapunov function, it may even restore the Lyapunov function below the performance barrier, thereby allowing the system to revert to P-CETC operation.

\end{rem}

\subsection{Performance-barrier Periodic Event-triggered Control (P-PETC)}
\label{subsec:P-PETC}

This section introduces the P-PETC approach, derived from the P-CETC scheme. This approach is realized by redesigning the triggering function, $\Gamma^p(t)$, of P-CETC to $\tilde{\Gamma}^p(t)$, allowing for periodic evaluation while ensuring that $\Gamma^p(t)$ remains non-positive and the dynamic variable $m^p(t)$ remains positive along the P-PETC closed-loop system solution.

Let $\tilde I^p=\left\{\tilde t_0^p, \tilde t_1^p, 
\tilde t_2^p, \ldots\right\}$ denote the sequence of event-times associated with P-PETC. Let the event-trigger parameters $\theta,\theta_m,\kappa_1,\kappa_2,\kappa_3>0$ be selected as outlined in Assumption \ref{asm:R-CETC-param}, let $\eta>0$ be chosen later, and $c>0$ be a design parameter. The proposed P-PETC strategy consists of two parts:
\begin{enumerate}
    \item [1)] An ETC input $\tilde U_k^p$
    \begin{equation}
    \label{eq:U-P-PETC}
         \begin{aligned}
            \tilde U_k^{p}:=U(\tilde t_k^{p}),
        \end{aligned}
    \end{equation}
    for $t \in\left[\tilde t^p_k, \tilde t^p_{k+1}\right), k \in \mathbb{N},$ where $U(t)$ is given by \eqref{eq:U-wbar-vbar}. Then, the boundary condition \eqref{eq:sys-wv-bd-cond-Uk} becomes
    \begin{equation}
    \label{eq:P-PETC-bd}
        \bar{v}(\ell, t) =r_1 \tilde U_k^p.
    \end{equation}
    \item [2)] A periodic event-trigger determining the event-times 
    \begin{equation}
    \label{eq:P-PETC-trigger-t}
        \begin{aligned}
        \tilde t^p_{k+1}=\inf \{t \in \mathbb{R}_{+} \mid  & t>\tilde t_k^p, \tilde \Gamma^p(t)>0 , t=nh, \\
        & h>0, n\in \mathbb{N},  k \in \mathbb{N} \},
        \end{aligned}
    \end{equation}
    with $\tilde t_0^p=0,$ $h$ is the sampling period satisfying \eqref{eq:h-def},
    and the triggering function $\tilde \Gamma^p(t)$ defined as
\begin{equation}\label{eq:P-PETC-trigger-func}
\begin{split}
\tilde{\Gamma}^p(t ):=&
(a+\theta\theta_m)e^{ah}d^{2}(t )-\theta\theta_m d^2(t )-\theta a m^p(t)\\&-\frac{ac}{\theta_m}e^{-ch}W^p(t),
\end{split}
\end{equation}
   where $a$ is given by \eqref{eq:a-def}, $d(t)$ is given by \eqref{eq:d(t)} for $t \in\left[\tilde{t}_k^p, \tilde{t}_{k+1}^p\right),k \in \mathbb{N}$, $m^p(t)$ is governed by \eqref{eq:ODE-dot-m-p}  along the solution of \eqref{eq:sys-ab-d-a}-\eqref{eq:sys-ab-d-bd-cond-b} for $t \in\left(\tilde{t}_k^p, \tilde{t}_{k+1}^p\right),k \in \mathbb{N}$, and  $W^p(t)$ is the performance residual given by \eqref{eq:W-p}-\eqref{eq:Lyap-V0}.
\end{enumerate}

Next we present Lemmas \ref{lem:P-PETC-d}-\ref{prop:P-PETC-Gamma}, required for proving the main results of P-PETC in Theorem \ref{thm:P-PETC}.

\begin{lem}
\label{lem:P-PETC-d}
    Consider the P-PETC approach \eqref{eq:U-P-PETC}-\eqref{eq:P-PETC-trigger-func}.
    For $d(t)$ given by \eqref{eq:d(t)}, it holds that
    \begin{align}
    \label{eq:dot-d(t)-est}
     (\dot{d}(t))^2\leq&\varepsilon_0d^2(t)+\varepsilon_1\Vert\alpha[t]\Vert^2+\varepsilon_2\Vert\beta[t]\Vert^2+\varepsilon_3\alpha^2(1,t),
    \end{align}
    along the solution of \eqref{eq:sys-ab-d-a}-\eqref{eq:sys-ab-d-bd-cond-b} for all $t \in$ $(n h,(n+1) h)$ and any $n \in\left[\tilde{t}_k^p / h, \tilde{t}_{k+1}^p / h\right) \cap \mathbb{N}$. Here $\varepsilon_0,\varepsilon_1,\varepsilon_2,\varepsilon_3>0$ are give by \eqref{eq:eps-2},\eqref{al1}-\eqref{gg3}, respectively.
\end{lem}

Lemma \ref{lem:P-PETC-d} is the direct result from Lemma 1 in \cite{espitia2020event} and thus the proof omitted.

\begin{lem}
\label{lem:P-PETC-eta}
    Consider the P-PETC approach \eqref{eq:U-P-PETC}-\eqref{eq:P-PETC-trigger-func}.
    Let parameter $\eta>0$ be chosen such that
    \begin{equation}
    \label{eq:eta=b*}
        \eta \leq b,
    \end{equation}
    where $b$ is given by \eqref{bbcfgj}. Then, the residual $W^p(t)$ given by \eqref{eq:W-p}-\eqref{eq:Lyap-V0} satisfies
    \begin{equation}
    \label{eq:Wp(nh)}
        W^p(t) \geq  e^{-\left(b^\star+c\right)(t-n h)} W^p(n h), \text{ with }b^\star = \eta,
    \end{equation}
    along the solution of \eqref{eq:sys-ab-d-a}-\eqref{eq:sys-ab-d-bd-cond-b},\eqref{eq:ODE-dot-m-p} for all $t \in[n h,(n+1) h)$ and any $n \in\left[\tilde{t}_k^p / h, \tilde{t}_{k+1}^p / h\right) \cap \mathbb{N}$. Further, it holds that
\begin{equation}\label{aaznml}
        W^p(t)\geq0 \text{  i.e.,  }e^{-b^\star t } V_0^p \geq V^p(t), \text{ with }b^\star = \eta,
    \end{equation}
for all $t>0$.
\end{lem}

\begin{proof}[\rm \textbf{Proof}]
    Differentiating \eqref{eq:Lyap-V-p} along the solution of \eqref{eq:sys-ab-d-a}-\eqref{eq:sys-ab-d-bd-cond-b},\eqref{eq:ODE-dot-m-p} in $t \in[n h,(n+1) h)$ and any $n \in\left[\tilde{t}_k^p / h, \tilde{t}_{k+1}^p / h\right) \cap \mathbb{N}$, we obtain
    \begin{align}
        \dot{V}^p(t)  
        & \leq -b V_1(t)-\eta m^p(t)+cW^p(t) \label{eq:eta=b*-proof}\\
        & = -\eta \left(V_1(t)+m^p(t)\right) + (\eta-b)V_1(t) + cW^p(t) . \nonumber
    \end{align}
   Selecting $\eta$ as in \eqref{eq:eta=b*}, we can get rid of $V_1(t)$ term to obtain $\dot{V}^p(t) \leq -\eta V^p(t) + cW^p(t)$. 
   Following the similar process \eqref{eq:Lyap-V-p-dot}-\eqref{eq:W-p-proof} in the proof of Theorem \ref{thm:P-CETC}, we obtain \eqref{eq:Wp(nh)} for all $t \in$ $[n h,(n+1) h)$ and any $n \in\left[\tilde{t}_k^p / h, \tilde{t}_{k+1}^p / h\right) \cap \mathbb{N}$ under the P-PETC approach \eqref{eq:U-P-PETC}-\eqref{eq:P-PETC-trigger-func}. Further, following similar arguments, the relation \eqref{aaznml} valid for all $t>0$ can be obtained, due to the absence of Zeno behavior under P-PETC. 
\end{proof}

\begin{lem}
\label{prop:P-PETC-Gamma}
    Consider the P-PETC approach \eqref{eq:U-P-PETC}-\eqref{eq:P-PETC-trigger-func} with the event-trigger parameters $\theta,\theta_m,\kappa_1,\kappa_2,\kappa_3>0$ selected as in Assumption \ref{asm:R-CETC-param}, $c>0$, and $\eta>0$ chosen as in \eqref{eq:eta=b*}. Then, $\Gamma^p(t)$ of P-CETC given by \eqref{eq:P-CETC-trigger-func} satisfies
    \begin{equation}\label{eq:P-PETC-Gp-est}
 \begin{split}
    &\Gamma^p(t)\\&\leq \frac{1}{a}\Big((a+\theta\theta_m)d^{2}(nh)e^{a(t-nh)}-\theta\theta_m d^2(nh)\\&\quad\qquad-\theta a m^p(nh)-\frac{ac}{\theta_m}e^{-c(t-nh)}W^p(nh)\Big)e^{-\eta(t-nh)},
\end{split}
\end{equation}
where $a$ is given by \eqref{eq:a-def}, and $h$ is the sampling period given by \eqref{eq:h-def}, along the solution of \eqref{eq:sys-ab-d-a}-\eqref{eq:sys-ab-d-bd-cond-b},\eqref{eq:ODE-dot-m-p} for all $t \in[n h,(n+1) h)$ and any $n \in\left[\tilde{t}_k^p / h, \tilde{t}_{k+1}^p / h\right) \cap \mathbb{N}$.
\end{lem}

\begin{proof}[\rm \textbf{Proof}]
     Since it was shown in Lemma \ref{lem:P-PETC-eta} that $ W^p(t) \geq  e^{-\left(b^\star+c\right)(t-n h)} W^p(n h)$ with $b^\star=\eta$ for all $t \in$ $[n h,(n+1) h)$ and any $n \in\left[\tilde{t}_k^p / h, \tilde{t}_{k+1}^p / h\right) \cap \mathbb{N}$, it follows that
    \begin{equation}\label{fgvcmt}
    \Gamma^p(t) \leq d^2(t) - \theta m^p(t) - \frac{c}{\theta_m} e^{-\left(\eta+c\right)(t-n h)} W^p(n h), 
    \end{equation}
    for all $t \in$ $[n h,(n+1) h)$ and any $n \in\left[\tilde{t}_k^p / h, \tilde{t}_{k+1}^p / h\right) \cap \mathbb{N}$, along the solution of \eqref{eq:sys-ab-d-a}-\eqref{eq:sys-ab-d-bd-cond-b},\eqref{eq:ODE-dot-m-p}. We define
    \begin{equation}
    \label{eq:P-PETC-Gp*}
        \Gamma^{p*}(t) := \ d^2(t) - \theta_m m^p(t).
    \end{equation}
 By taking the time derivative of \eqref{eq:P-PETC-Gp*} in $t \in(n h,(n+1) h)$ and $n \in\left[\tilde{t}_k^p / h, \tilde{t}_{k+1}^p / h\right) \cap \mathbb{N}$, using Young's inequality, substituting the estimation of $(\dot{d}(t))^2$ given by \eqref{eq:dot-d(t)-est}, and substituting $\dot {m}^p(t)$ given by \eqref{eq:ODE-dot-m-p}, we get
    \begin{equation}
    \label{eq:P-PETC-Gp*-proof-1}
    \begin{aligned}
    &\dot{\Gamma}^{p *}(t)  =  2 d(t) \dot{d}(t) \!-\! \theta_m \dot{m}^p(t) 
    \leq   d^2(t)\!+\!(\dot{d}(t))^2\!-\!\theta_m \dot{m}^p(t) \\
    & \leq \left(1+\varepsilon_0+\theta\theta_m\right) d^2(t)+ \theta\eta m^p(t)
     -(\theta\kappa_1-\varepsilon_1)\Vert\alpha[t]\Vert^2\\&-(\theta\kappa_2-\varepsilon_2)\Vert\beta[t]\Vert^2-(\theta\kappa_3-\varepsilon_3)\alpha^2(\ell,t)-\theta cW^p(t),
    \end{aligned}
    \end{equation}
    where $W^p(t)$ is given by \eqref{eq:W-p}. Replacing $d^2(t)$ in \eqref{eq:P-PETC-Gp*-proof-1} with ${\Gamma}^{p *}(t)$ given by \eqref{eq:P-PETC-Gp*}, we further obtain
 \begin{equation}
    \label{eq:P-PETC-Gp*-proof-2}
    \begin{aligned}
    \dot{\Gamma}^{p *}(t) \leq &(1+\varepsilon_0+\theta\theta_m)\Gamma^{p*}(t)+\theta(a+\theta\theta_m)m^p(t)\\&
     -(\theta\kappa_1-\varepsilon_1)\Vert\alpha[t]\Vert^2-(\theta\kappa_2-\varepsilon_2)\Vert\beta[t]\Vert^2\\&-(\theta\kappa_3-\varepsilon_3)\alpha^2(\ell,t)-\theta cW^p(t).
    \end{aligned}
    \end{equation}
    It can be shown that both sides of \eqref{eq:P-PETC-Gp*-proof-2} are well-behaved in $t \in(n h,(n+1) h)$ and $n \in\left[\tilde{t}_k^p / h, \tilde{t}_{k+1}^p / h\right) \cap \mathbb{N}$. Therefore, there exists a nonnegative function $\iota(t) \in C^0\left(\left(\tilde{t}_k^p, \tilde{t}_{k+1}^p\right) ; \mathbb{R}_{+}\right)$ such that
     \begin{equation}
    \label{eq:P-PETC-Gp*-proof-3}
    \begin{aligned}
    \dot{\Gamma}^{p *}(t) = &(1+\varepsilon_0+\theta\theta_m)\Gamma^{p*}(t)+\theta(a+\theta\theta_m)m^p(t)\\&
     -(\theta\kappa_1-\varepsilon_1)\Vert\alpha[t]\Vert^2-(\theta\kappa_2-\varepsilon_2)\Vert\beta[t]\Vert^2\\&-(\theta\kappa_3-\varepsilon_3)\alpha^2(\ell,t)-\theta cW^p(t)-\iota(t),
    \end{aligned}
    \end{equation}
    for all $t \in(n h,(n+1) h)$ and $n \in\left[\tilde{t}_k^p / h, \tilde{t}_{k+1}^p / h\right) \cap \mathbb{N}$. Furthermore, replacing $d^2(t)$ with ${\Gamma}^{p *}(t)$ term given by \eqref{eq:P-PETC-Gp*}, we can rewrite the dynamics of $m^p(t)$ given by \eqref{eq:ODE-dot-m-p} as
    \begin{equation}\label{eq:P-PETC-Gp*-proof-dot-m-p}
\begin{split}
    \dot{m}^p(t) =& -\theta_m \Gamma^{p*}(t )-(\theta\theta_m+\eta)m^p(t)+\kappa_1\Vert \alpha[t]\Vert^2\\&+\kappa_2 \Vert \beta[t]\Vert^2+\kappa_3\alpha^2(\ell,t)+cW^p(t),
\end{split}
\end{equation}
for $t \in(n h,(n+1) h)$ and $n \in\left[\tilde{t}_k^p / h, \tilde{t}_{k+1}^p / h\right) \cap \mathbb{N}$. Then, combining \eqref{eq:P-PETC-Gp*-proof-3} with \eqref{eq:P-PETC-Gp*-proof-dot-m-p}, we can obtain the following ODE system
    \begin{equation}
    \label{eq:z-dot}
    \dot{z}(t)=A z(t)+\nu(t),
    \end{equation}
    where
    \begin{equation}
    \begin{aligned}
     z(t)\!&=\!\left[\!\begin{array}{l}
    \! \Gamma^{p *} (t)\! \\
    \! m^p (t)\!
    \end{array}\!\right]\!,\\ 
    A &= \begin{bmatrix}
1+\varepsilon_0+\theta\theta_m &  \theta\big(a+\theta\theta_m\big)\\
    -\theta_m & -(\theta\theta_m+\eta)
    \end{bmatrix} , \\
    \nu(t )  &=\begin{bmatrix}
    \Big(
    \begin{split}
        &-(\theta\kappa_1-\varepsilon_1)\Vert \alpha[t]\Vert^2-(\theta\kappa_2-\varepsilon_2)\Vert\beta[t]\Vert^2\\&-(\theta\kappa_3-\varepsilon_3)\alpha^2(\ell,t)-\theta cW^p(t)-\iota(t )
        \end{split}\Big)\\\kappa_1\Vert \alpha[t]\Vert^2+\kappa_2\Vert \beta[t]\Vert^2+\kappa_3\alpha^2(\ell,t)+cW^p(t)
    \end{bmatrix} . \\
    \end{aligned}
    \end{equation}
    Solving \eqref{eq:z-dot} gives us
    \begin{equation}
    z(t)=e^{A(t-n h)} z(n h)+\int_{n h}^t e^{A(t-\xi)} \nu(\xi) d \xi,
    \end{equation}
    for all $t \in[n h,(n+1) h)$ and $n \in$ $\left[\tilde{t}_k^p / h, \tilde{t}_{k+1}^p / h\right) \cap \mathbb{N}$.
    Then we can obtain that
    \begin{equation}
    \label{eq:P-PETC-Gp*-proof-4}
    \Gamma^{p *}(t)=D e^{A(t-n h)} z(n h)+\int_{n h}^t D e^{A(t-\xi)} \nu(\xi) d \xi, 
    \end{equation}
    where $D=[1\quad 0]$. The matrix $A$ has two distinct eigenvalues $-\eta$ and $1+\varepsilon_0$. To find an upper bound of $\Gamma^{p *}(t)$, we diagonalize the matrix exponential $e^{A t}$ as follows:
\begin{equation}\label{expm}
\begin{split}
    e^{At}=\frac{\theta_m}{a}\begin{bmatrix}
        -\theta & -\frac{a+\theta\theta_m}{\theta_m}\\1&1
    \end{bmatrix}&\begin{bmatrix}
        e^{-\eta t}&0\\0& e^{\big(1+\varepsilon_0\big)t}
    \end{bmatrix}\begin{bmatrix}
        1 & \frac{a+\theta\theta_m}{\theta_m}\\-1&-\theta
    \end{bmatrix}.
\end{split}
\end{equation}
Then, we can show that
\begin{equation}\label{eq:proof-Cev}
\begin{split}
    &De^{A(t-\xi)}\nu(\xi )\\&=-\Big(\big(\theta\kappa_1-\varepsilon_1\big)g_1(t-\xi)-\kappa_1 g_2(t-\xi)\Big)\Vert \alpha[\xi]\Vert^2\\&\quad-\Big(\big(\theta\kappa_2-\varepsilon_2\big)g_1(t-\xi)-\kappa_2 g_2(t-\xi)\Big)\Vert\beta[t]\Vert^2\\&\quad-\Big(\big(\theta\kappa_3-\varepsilon_3\big)g_1(t-\xi)-\kappa_3 g_2(t-\xi)\Big)\alpha^2(\ell,\xi)\\& \quad-c\Big(\theta  g_1(t-\xi)-g_2(t-\xi)\Big)W^p(\xi)
    \\&\quad-g_1(t-\xi)\iota(\xi ),
\end{split}
\end{equation}
where
\begin{equation} \label{eq:P-PETC-g1(t)}
    g_1(t)=\frac{1}{a}\Big(-\theta\theta_m +(a+\theta\theta_m)e^{at}\Big)e^{-\eta t},
\end{equation}
\begin{equation}\label{eq:P-PETC-g2(t)}
      g_2(t)=\frac{\theta (a+\theta\theta_m)}{a}\Big(-1+ e^{at}\Big)e^{-\eta t}.
\end{equation}
    Noting that $a>0$, it is obvious that $g_1(t)>0, g_2(t)>0$  and
\begin{equation}
    \theta g_1(t)- g_2(t) = \theta e^{-\eta t}>0,
\end{equation}
for all $t \geq 0$.  Also, noting that $\theta\kappa_i/\varepsilon_i=1/(1-\sigma), i=1,2,3$ from \eqref{betas}, and recalling \eqref{eq:tau_d}, we obtain that
\begin{equation}
\begin{split}
 &\big(\theta\kappa_i-\varepsilon_i\big)g_1(t-\xi)-\kappa_i g_2(t-\xi)\\&= \frac{\varepsilon_i (a+\theta\theta_m) }{a}\bigg(1+\frac{\sigma a}{(1-\sigma)(a+\theta\theta_m)}-e^{a(t-\xi)}\bigg)e^{-\eta (t-\xi)}\\&=
 \frac{\varepsilon_i (a+\theta\theta_m) }{a}\bigg(e^{a\tau_d}-e^{a(t-\xi)}\bigg)e^{-\eta (t-\xi)},
\end{split}
\end{equation}
for all $i=1,2,3$. As $n h \leq \xi \leq t<(n+1) h$, and $h \leq \tau_d$, we have $e^{a \tau_d}-e^{a(t-\xi)}>0$ and thus $\big(\theta\kappa_i-\varepsilon_i\big)g_1(t-\xi)-\kappa_i g_2(t-\xi)>0$ for all $i=1,2,3$. As a result, every term in \eqref{eq:proof-Cev} is non-positive and we can argue that $D e^{A(t-\xi)} \nu(\xi) \leq 0$ for all $t, \xi$ such that $n h \leq \xi \leq t<(n+1) h$, and $n \in$ $\left[\tilde{t}_k^p / h, \tilde{t}_{k+1}^p / h\right) \cap \mathbb{N}$. Considering this fact along with \eqref{eq:P-PETC-Gp*-proof-4}, we obtain that for $t \in[n h,(n+1) h)$
\begin{equation}\label{eq:P-PETC-Gp*-proof-5}
    \begin{split}
        &\Gamma^{p*}(t )\leq De^{A(t-nh)}z(nh )\\
&\leq      g_1(t-nh)\Gamma^{p*}(nh )+g_2(t-nh)m^p(nh)\\&\leq 
\frac{1}{a}\Big(-\theta(a+\theta\theta_m)m^p(nh)-\theta\theta_m\Gamma^{p*}(nh )\\&+(a+\theta\theta_m)\big(\Gamma^{p*}(nh )+\theta m^p(nh)\big)e^{a(t-nh)}\Big)e^{-\eta(t-nh)}.
    \end{split}
\end{equation}
By using \eqref{eq:P-PETC-Gp*} to eliminate $\Gamma^{p *}(n h)$ on the R.H.S. of \eqref{eq:P-PETC-Gp*-proof-5}, we obtain
 \begin{equation}\label{eq:P-PETC-Gp*-proof-6}
 \begin{split}
    &\Gamma^{p*}(t)\\&\leq \frac{1}{a}\Big((a+\theta\theta_m)d^{2}(nh)e^{a(t-nh)}-\theta\theta_m d^2(nh)\\&\quad\qquad-\theta a m^p(nh)\Big)e^{-\eta(t-nh)},
\end{split}
\end{equation}
Then, recalling \eqref{fgvcmt} and \eqref{eq:P-PETC-Gp*}, and using \eqref{eq:P-PETC-Gp*-proof-6}, we can obtain the inequality \eqref{eq:P-PETC-Gp-est} for $t \in[n h,(n+1) h)$, which completes the proof.
\end{proof}

Now we state the main results of P-PETC below.
\begin{thm}[Results under P-PETC]
    \label{thm:P-PETC}
    Let $\tilde I^p=$ $\left\{\tilde t_0^p, \tilde t_1^p, \tilde t_2^p, \ldots\right\}$ with $\tilde t_0^p = 0$ be the set of increasing event-times generated by the P-PETC approach \eqref{eq:U-P-PETC}-\eqref{eq:P-PETC-trigger-func} with appropriate choices for the event-trigger parameters under Assumption \ref{asm:R-CETC-param}, $c>0$, and $\eta>0$ chosen as in \eqref{eq:eta=b*}. 
    Then, the following results hold:
    \begin{enumerate}
        
        \item[R1:] For every $\left(\bar w (\cdot, 0), \bar v (\cdot, 0)\right)^T \in L^2\left((0,\ell); \mathbb{R}^2\right)$, there exists a unique solution $(\bar w, \bar v)^T \in \mathcal{C}^0\left(\mathbb{R}_+ ; L^2\left((0,\ell) ; \mathbb{R}^2\right)\right)$ to the P-PETC closed-loop system \eqref{eq:sys-wv-wbar}-\eqref{eq:sys-wv-bd-cond-w},\eqref{cxx}-\eqref{cndtns},\eqref{eq:U-P-PETC}-\eqref{eq:P-PETC-trigger-func} for all $t>0$.

        \item[R2:] The function $\Gamma^p(t)$ given by \eqref{eq:P-CETC-trigger-func} satisfies $\Gamma^p(t) \leq 0$
        for all $t>0$, along the P-PETC closed-loop solution. 

        \item[R3:] The dynamic variable $m^p(t)$ governed by \eqref{eq:ODE-dot-m-p} with $m^p(0)=m^r(0)>0$ satisfies $m^p(t)>0$
        for all $t > 0$, along the P-PETC closed-loop solution.

        \item[R4:]The Lyapunov candidate $V^p(t)$ given by \eqref{eq:Lyap-V-p} satisfies \eqref{zzzbnmkfb} for all $t\in (\tilde{t}_k^p,\tilde{t}_{k+1}^p),k\in\mathbb{N}$ and \eqref{eq:Lyap-V-p<ebt} for all $t>0$ with  $\eta = b^\star$, along the P-PETC closed-loop solution.
        
        \item[R5:]The closed-loop signal $\Vert\bar{w}[t]\Vert+\Vert\bar{v}[t]\Vert$ associated with the P-PETC closed-loop system \eqref{eq:sys-wv-wbar}-\eqref{eq:sys-wv-bd-cond-w},\eqref{cxx}-\eqref{cndtns},\eqref{eq:U-P-PETC}-\eqref{eq:P-PETC-trigger-func}, exponentially converges to zero.
    \end{enumerate}
\end{thm}

\begin{proof}[\rm \textbf{Proof}]
      R1 follows from Proposition \ref{prop:wellpose} and Remark \ref{rem:wellpose}. Lemma \ref{lem:P-PETC-eta} ensures that $W^p(t) \geq 0$ for all $t > 0$ with properly chosen parameters listed in Assumption \ref{asm:R-CETC-param} and $\eta = b^\star$.
    Consider the interval $t \in\left[\tilde{t}_k^p, \tilde{t}_{k+1}^p\right)$.
    Assume that an event has triggered at $t=\tilde{t}_k^p$ and $m^p\left(\tilde{t}_k^p\right)>0$. At $t=\tilde{t}_k^p$, the control law is updated, so $d(\tilde{t}_k^p)=0$. 
        Then we have from \eqref{eq:P-CETC-trigger-func} that 
        \begin{equation}
            \Gamma^p\left(\tilde{t}_k^p\right)=- \theta m^p(\tilde{t}_k^p)-\frac{c}{\theta_m}W^p(\tilde{t}_k^p)<0.
        \end{equation}
        Then, $\Gamma^p(t)$ will at least remain non-positive until $t=\tilde{t}_k^p+\tau_d$, where $\tau_d$ is the minimal dwell-time of P-CETC given by R2 of Theorem \ref{thm:P-CETC}.
        Since $h \leq \tau_d$, $\Gamma^p(t)$ will definitely remain non-positive in $t \in\left[\tilde{t}_k^p, \tilde{t}_k^p+h\right)$.
        At each $t=n h, n>0, n \in \mathbb{N}$, the P-PETC given by \eqref{eq:U-P-PETC}-\eqref{eq:P-PETC-trigger-func} is evaluated, and only when $\tilde{\Gamma}^p(n h)>0$ that an event is triggered and the control input is updated.
        When $\tilde{\Gamma}^p(n h) \leq 0$, \eqref{eq:P-PETC-trigger-func} and \eqref{eq:P-PETC-Gp-est} imply the right hand side of \eqref{eq:P-PETC-Gp-est} is non-positive and thus $\Gamma^p(t)$ will definitely remain non-positive at least until $t=\tilde{t}_{k+1}^p$ when $\tilde{\Gamma}^p\left(\tilde{t}_{k+1}^{p-}\right)>0$.
        Since $\Gamma^p(t)\leq0$ for $t \in\left[\tilde{t}_k^p, \tilde{t}_{k+1}^p\right)$, we follow a process similar to the proof of Lemma \ref{lem:P-CETC-m} and obtain $m^p(t)>0$ for $t \in\left[\tilde{t}_k^p, \tilde{t}_{k+1}^p\right]$.
        Therefore, after the control input has been updated at $t=\tilde{t}_{k+1}^p$, we have $\Gamma^p\left(\tilde{t}_{k+1}^p\right)= -\theta m^p\left(\tilde{t}_{k+1}^p\right)-\frac{c}{\theta_m}W^p(\tilde{t}_{k+1}^p) < 0$.
        Similarly, we can analyze the behavior of $\Gamma^p(t)$ and $m^p(t)$ in all $t \in\left[\tilde{t}_k^p, \tilde{t}_{k+1}^p\right)$ for any $k \in \mathbb{N}$ starting from $\tilde{t}_0^p=0$ and $m^p(0)>0$ to prove that $\Gamma^p(t) \leq 0$ for all $t \in\left[\tilde{t}_k^p, \tilde{t}_{k+1}^p\right), k \in \mathbb{N}$ and $m^p(t)>0$ for all $t >0$, as stated in R2 and R3. As $m^p(t)>0$ for all $t >0$ guarantees the positive definiteness of $V^p(t)$, we have R4 and R5 by following similar arguments in the proofs of R4 and R5 of Theorem \ref{thm:P-CETC}.
\end{proof}

\subsection{Performance-barrier Self-triggered Control (P-STC)}
\label{subsec:P-STC}

In this subsection, we introduce the P-STC approach derived from the P-CETC scheme. P-STC determines the next event time at the current event time using continuously available measurements, a prediction of the closed-loop system states, and bounds of the constituent terms of the P-CETC triggering function $\Gamma^p(t)$. We show the P-STC approach ensures that $\Gamma^p(t)$ given by \eqref{eq:P-CETC-trigger-func} remains non-positive, and $m^p(t)$ given by \eqref{eq:ODE-dot-m-p} remains positive along the P-STC closed-loop system solution.

Let $\check I^p=\left\{\check t_0^p, \check t_1^p, \check t_2^p, \ldots\right\}$ denote the sequence of event-times associated with P-STC. Let the event-trigger parameters $\theta,\theta_m,\kappa_1,\kappa_2,\kappa_3>0$ be selected as outlined in Assumption \ref{asm:R-CETC-param}, let $\eta>0$ be chosen as in \eqref{eq:eta=b*}, and let $c>0$ be a design parameter. The proposed P-STC strategy consists of two parts:
\begin{enumerate}
    \item An event-triggered boundary control input $\check U_k^p$
    \begin{equation}
        \label{eq:U-P-STC}
            \begin{aligned}
                \check U_k^p := U(\check t_k^p),
            \end{aligned}
    \end{equation}
    for $t \in\left[\check t^p_k, \check t^p_{k+1}\right), k \in \mathbb{N}$ where $U(t)$ is given by \eqref{eq:U-wbar-vbar}. Then, the boundary condition \eqref{eq:sys-wv-bd-cond-Uk} becomes
    \begin{equation}
    \label{eq:P-STC-bd}
        \bar{v}(\ell, t) =r_1 \check U_k^p.
    \end{equation}
    \item A self-trigger determining event-times
    \begin{equation}
    \label{eq:P-STC-trigger-t} \check{t}_{k+1}^p=\check{t}_k^p+G^p\left(H(\check{t}_k^p), m^p\left(\check{t}_k^p\right), W^p(\check{t}_k^p)\right),
    \end{equation}
    with $\check t_0^p=0$ and $G^p(\cdot, \cdot, \cdot)>0$ being a positively and uniformly lower-bounded function
    \begin{equation}
        \begin{aligned}
            \label{eq:P-STC-trigger-func}
            G^p\left( H(t), m^p(t), W^p(t) \right)
            :=\max \left\{\tau_d, \check \tau (t) \right\} ,
        \end{aligned}
    \end{equation}
    where
    \begin{equation}
    \label{eq:tau_check}
        \check \tau (t) = \frac{1}{\varrho^\star+\eta+c}\ln\Big(\frac{\theta m^p(t)+\frac{\theta\theta_m H(t)}{\varrho^\star+\eta}+\frac{c}{\theta_m}W^p(t)}{H(t)+\frac{\theta\theta_mH(t)}{\varrho^\star+\eta}}\Big)
    \end{equation}
    and $\tau_d$ is the R/P-CETC minimum dwell-time given by \eqref{eq:tau_d}-\eqref{eq:eps-2}. The variable $m^p(t)$ satisfies the dynamics \eqref{eq:ODE-dot-m-p} along the solution of \eqref{eq:d(t)}-\eqref{eq:sys-ab-d-bd-cond-b} for $t \in\left(\check{t}_k^p, \check{t}_{k+1}^p\right),k \in \mathbb{N}$, $W^p(t)$ is the performance residual given by \eqref{eq:W-p}-\eqref{eq:Lyap-V0}, and $H(t)$ is given by \eqref{zmlsbnj}.
\end{enumerate}

Next we present Lemma \ref{lem:P-STC-dvm} which provides bounds on $d^2(t)$, $m^p(t)$, and $W^p(t)$ required for proving the main results of P-STC.

\begin{lem}
\label{lem:P-STC-dvm} Consider the P-STC approach \eqref{eq:U-P-STC}-\eqref{eq:tau_check}, which generates an increasing set of event times $\{\check{t}^p_{k}\}_{k\in\mathbb{N}}$ with $\check{t}^p_{k}=0$. Then, for the input holding error error $d(t)$ given by  \eqref{eq:d(t)}, the following estimate holds
\begin{equation}\label{bbv111d}
    d^2(t)\leq H(\check{t}^p_k)e^{\varrho^\star(t-\check{t}^p_k)},
\end{equation}
where $\varrho^\star>0$ is given by \eqref{dxxxml}, and $H(t)$ is given by \eqref{zmlsbnj}. Further, if the event-trigger parameters $\theta,\theta_m,\kappa_1,\kappa_2,\kappa_3>0$ are chosen as in Assumption \ref{asm:R-CETC-param}, $c>0$, and $\eta>0$ is chosen as in \eqref{eq:eta=b*},  then $W^p(t)$ given by \eqref{eq:W-p}-\eqref{eq:Lyap-V0} satisfies
\begin{equation}\label{akpgh}
    W^p(t)\geq e^{-(b^{\star}+c)(t-\check{t}_k^p)}W^p(\check{t}_k^p)\text{ with }b^{\star}=\eta,
\end{equation}
 for all $t\in\big[\check{t}_k^p,\check{t}_{k+1}^p\big),k\in\mathbb{N},$ and
\begin{equation}\label{akpqgh}
   W^p(t)\geq 0,\text{ i.e., } V^p(t)\leq e^{-b^{\star}t}V_0\text{ with }b^{\star}=\eta,
\end{equation}
for all $t>0$ whereas $m^p(t)$ governed by \eqref{eq:ODE-dot-m-p} satisfies
\begin{equation}\label{hjbbv211}
\begin{split}
    m^p(t)\geq& m^p(\check{t}_k^p)e^{-\eta(t-\check{t}_k^p)}\\&-\frac{\theta_m H(\check{t}_k^p)}{\varrho^\star+\eta}e^{-\eta (t-\check{t}_k^p)}\Big(e^{(\varrho^\star+\eta)(t-\check{t}_k^p)}-1\Big),
\end{split}
\end{equation}
 for all $t\in[\check{t}^p_{k},\check{t}^p_{k+1}),j\in\mathbb{N}$.
\end{lem}
\begin{proof}[\rm \textbf{Proof}] Let us consider the following positive definite function
\begin{align}\label{dcmlspt}
    \bar{V}(t)\hspace{-3pt}:=& \hspace{-2pt}\int_{0}^{\ell}\hspace{-5pt}\Big(\frac{1}{v^{\star}}\alpha^2(x,t)e^{-\frac{\mu x}{v^{\star}}}+\frac{r_0^2}{(\gamma p^{\star}-v^{\star})}\beta^2(x,t)e^{\frac{\mu x}{(\gamma p^{\star}-v^{\star})}}\Big)dx,
\end{align}
\noindent where $\mu>0$. Taking the time derivative of \eqref{dcmlspt}, integrating by parts, and recalling \eqref{eq:sys-ab-d-bd-cond-b}, we can obtain
\begin{align}\label{dfg1df}
\begin{split}
    \dot{\bar{V}}(t) =& -\mu \bar{V}(t)-e^{-\frac{\mu\ell}{v^\star}}\alpha^2(\ell,t)+r_0^2r_1^2e^{\frac{\mu\ell}{(\gamma p^{\star}-v^{\star})}}d^2(t)\\\leq & r_0^2r_1^2e^{\frac{\mu\ell}{(\gamma p^{\star}-v^{\star})}}d^2(t)
\end{split}
\end{align}
for all $t\in(\check{t}^p_k,\check{t}^p_{k+1}),j\in\mathbb{N}$. Using Young's and Cauchy–Schwarz inequalities on \eqref{eq:d(t)} and recalling \eqref{eq:tilde-L-21-L22} and \eqref{dcmlspt}, we obtain that
\begin{align}
\begin{split}\label{kkfmxlprt}
    &d^2(t) \leq \frac{4\tilde{L}^{21}}{r_1^2}\int_{0}^{\ell}\alpha^2(y,\check{t}^p_k)dy+\frac{4\tilde{L}^{22}}{r_1^2} \int_{0}^{\ell}\beta^2(y,\check{t}^p_k)dy\\&\qquad\quad+\frac{4\tilde{L}^{21}}{r_1^2} \int_{0}^{\ell}\alpha^2(y,t)dy+\frac{4\tilde{L}^{22}}{r_1^2} \int_{0}^{\ell}\beta^2(y,t)dy\\&\leq \frac{4v^\star\tilde{L}^{21}e^{\frac{\mu \ell}{v^\star}}}{r_1^2}  \int_{0}^{\ell}\frac{1}{v^\star}\alpha^2(y,\check{t}^p_k)e^{-\frac{\mu y}{v^\star}}dy\\&\quad+\frac{4 (\gamma p^{\star}-v^{\star}) \tilde{L}^{22}}{r_0^2r_1^2}\int_{0}^{\ell}\frac{r_0^2}{(\gamma p^{\star}-v^{\star})}\beta^2(y,\check{t}^p_k) e^{\frac{\mu y}{(\gamma p^{\star}-v^{\star})}}dy\\ &\quad+\frac{4v^\star\tilde{L}^{21}e^{\frac{\mu\ell}{v^\star}}}{r_1^2}  \int_{0}^{\ell}\frac{1}{v^\star}\alpha^2(y,t)e^{-\frac{\mu y}{v^\star}}dy\\&\quad+\frac{4 (\gamma p^{\star}-v^{\star}) \tilde{L}^{22}}{r_0^2r_1^2}\int_{0}^{\ell}\frac{r_0^2}{(\gamma p^{\star}-v^{\star})}\beta^2(y,t) e^{\frac{\mu y}{(\gamma p^{\star}-v^{\star})}}dy\\&\leq \varrho \bar{V}(\check{t}^p_k)+\varrho \bar{V}(t),
\end{split}
\end{align}
for all $t\in(\check{t}^p_k,\check{t}^p_{k+1}),j\in\mathbb{N},$ where $\varrho$ is given by \eqref{zzzmlw2e2}. Thus, it follows from \eqref{dfg1df} that
\begin{align}
    \dot{\bar{V}}(t) \leq \varrho^{\star} \bar{V}(\check{t}^p_k)+\varrho^\star \bar{V}(t),
\end{align}
for all $t\in(\check{t}^p_k,\check{t}^p_{k+1}),j\in\mathbb{N}$, where $\varrho^\star$ is given by \eqref{dxxxml}, from which we obtain
\begin{align}
\begin{split}
    \bar{V}(t)&\leq 2\bar{V}(\check{t}_k^p)e^{\varrho^\star(t-\check{t}^p_k)}-\bar{V}(\check{t}^p_k)\\&\leq 2\bar{V}(\check{t}^p_k)e^{\varrho^\star(t-\check{t}^p_k)},
\end{split}
\end{align}
for all $t\in [\check{t}^p_k,\check{t}^p_{k+1}],j\in\mathbb{N}$. Therefore, considering \eqref{kkfmxlprt}, we obtain
\begin{align}
    d^2(t)&\leq \varrho V(\check{t}^p_k)+2\varrho V(\check{t}^p_k)e^{\varrho^\star (t-\check{t}^p_k)}\\&\leq 3\varrho V(\check{t}^p_k)e^{\varrho^\star (t-\check{t}^p_k)},
\end{align}
which leads to \eqref{bbv111d}. Similar to Lemma \ref{lem:P-PETC-eta}, we can show that $W^p(t)\geq0, i.e., e^{-b^\star t} V_0^p \geq V^p(t)$ for any $t>0$ with $\eta=b^\star$ under the P-STC approach \eqref{eq:U-P-STC}-\eqref{eq:tau_check}. Thus, Considering the dynamics of $m^p(t)$ given by \eqref{eq:ODE-dot-m-p} and the relation \eqref{bbv111d}, we can show
\begin{equation}
\dot{m}^p(t)\geq -\eta m^p(t)-\theta_m H(\check{t}_k^p)e^{\varrho^\star(t-\check{t}_k^p)},
\end{equation}
for $t \in\left(\check{t}_k^p, \check{t}_{k+1}^p\right), k \in \mathbb{N}$ from which we can obtain \eqref{hjbbv211} using the Comparison principle.
\end{proof}

\begin{thm}[Results under P-STC]
    \label{thm:P-STC}
    Let $\check I^p= \left\{\check t_0^p, \check t_1^p, \check t_2^p, \ldots\right\}$ with $\check t_0^p = 0$ be the set of increasing event-times generated by the P-STC approach \eqref{eq:U-P-STC}-\eqref{eq:tau_check} with appropriate choices for the event-trigger parameters under Assumption \ref{asm:R-CETC-param}, $c>0$, and $\eta>0$ chosen as in \eqref{eq:eta=b*}. 
    Then, the following results hold:
    \begin{enumerate}
        \item[R1:] For every $\left(\bar w (\cdot, 0), \bar v (\cdot, 0)\right)^T \in L^2\left((0,\ell); \mathbb{R}^2\right)$, there exists a unique solution $(\bar w, \bar v)^T \in \mathcal{C}^0\left(\mathbb{R}_+ ; L^2\left((0,\ell) ; \mathbb{R}^2\right)\right)$ to the P-STC closed-loop system \eqref{eq:sys-wv-wbar}-\eqref{eq:sys-wv-bd-cond-w},\eqref{cxx}-\eqref{cndtns},\eqref{eq:U-P-STC}-\eqref{eq:tau_check} for all $t>0$.

        \item[R2:] The function $\Gamma^p(t)$ given by \eqref{eq:P-CETC-trigger-func} satisfies $\Gamma^p(t) \leq 0$
        for all $t>0$, along the P-STC closed-loop solution. 

        \item[R3:] The dynamic variable $m^p(t)$ governed by \eqref{eq:ODE-dot-m-p} with $m^p(0)=m^r(0)>0$ satisfies $m^p(t)>0$
        for all $t > 0$, along the P-STC closed-loop solution. 

        \item[R4:] The Lyapunov candidate $V^p(t)$ given by \eqref{eq:Lyap-V-p} satisfies \eqref{zzzbnmkfb} for all $t\in (\check{t}_k^p,\check{t}_{k+1}^p),k\in\mathbb{N}$ and \eqref{eq:Lyap-V-p<ebt} for all $t>0$ with  $\eta = b^\star$, along the P-STC closed-loop solution.

        \item[R5:] The closed-loop signal $\Vert\bar{w}[t]\Vert+\Vert\bar{v}[t]\Vert$ associated with the P-STC closed-loop system \eqref{eq:sys-wv-wbar}-\eqref{eq:sys-wv-bd-cond-w},\eqref{cxx}-\eqref{cndtns},\eqref{eq:U-P-STC}-\eqref{eq:tau_check}, exponentially converges to zero.
    \end{enumerate}
\end{thm}
\begin{proof}[\rm \textbf{Proof}]
    R1 follows from Proposition \ref{prop:wellpose} and Remark \ref{rem:wellpose}. Lemma \ref{lem:P-STC-dvm} ensures that $W^p(t)  \geq 0$ for all $t > 0$ under the chosen parameters listed in Assumption \ref{asm:R-CETC-param}, $c>0$ and $\eta = b^\star$. Assume that an event has triggered at $t=\check {t}_k^p$ and $m^p\left(\check{t}_k^p\right)>0$. 
    Then, let us analyze the behavior of $\Gamma^p(t)$ in $t \in\left[\check{t}_k^p, \check{t}_{k+1}^p\right)$ along the solution of \eqref{eq:sys-wv-wbar}-\eqref{eq:sys-wv-bd-cond-w},\eqref{cxx}-\eqref{cndtns},\eqref{eq:U-P-STC}-\eqref{eq:tau_check}. After the event at $t=\check{t}_k^p$, the control law is updated, and $d(\check{t}^p_k)=0$. Then we have from \eqref{eq:P-CETC-trigger-func} that 
    \begin{equation}
        \Gamma^p\left(\check{t}_k^p\right)= -\theta m\left(\check{t}_k^p\right)- \frac{c}{\theta_m} W^p(\check{t}_k^p)<0.
    \end{equation}
    Consequently, $\Gamma^p(t)$ will definitely remain non-positive until $t=\check{t}_k^p+\tau_d$, where $\tau_d$ is the R/P-CETC minimal dwell-time given by \eqref{eq:tau_d}. Further, recalling \eqref{akpgh} and \eqref{hjbbv211}, we can obtain
    \begin{align}
    \begin{split}
    & \theta m^p(t) + \frac{c}{\theta_m} W^p(t) \\
    &\geq \theta m^p(\check{t}_k^p)e^{-(\eta+c)(t-\check{t}_k^p)}+\frac{c}{\theta_m}e^{-(\eta+c)(t-\check{t}_k^p)}W^p(\check{t}_k^p)\\&\quad -\frac{\theta\theta_m H(\check{t}_k^p)}{\varrho^\star+\eta}e^{\varrho^\star(t-\check{t}_k^p)}+\frac{\theta\theta_m H(\check{t}_k^p)}{\varrho^\star+\eta}e^{-(\eta+c)(t-\check{t}_k^p)}, \\
    &:=  F(t-\check{t}_k^p) \label{eq:P-STC-bd-mp-2},
    \end{split}
    \end{align}
    for $t \in\left[\check{t}_k^p, \check{t}_{k+1}^p\right), k \in \mathbb{N}$. Suppose there exists a positive solution $t^{\dagger}>\check{t}_k^p$ such that 
    \begin{equation}
    \label{eq:P-STC-bd-mp-3}
        \begin{aligned}
 H(\check{t}^p_k)e^{\varrho^\star(t^\dagger-\check{t}^p_k)} =  F(t^\dagger-\check{t}_k^p).
    \end{aligned}
    \end{equation}
    From \eqref{bbv111d}, we know the L.H.S of \eqref{eq:P-STC-bd-mp-3} is an increasing upper-bound for $d^2(t)$, and from \eqref{eq:P-STC-bd-mp-2}, the R.H.S of \eqref{eq:P-STC-bd-mp-3} is a decreasing lower-bound for $\theta m^p(t) +  \frac{c}{\theta_m} W^p(t)$.
    Consequently, we can be certain that $ d^2(t)\leq \theta m^p(t) +  \frac{c}{\theta_m} W^p(t)$, i.e., $\Gamma^p(t) \leq 0$ for $t \in\left[\check{t}_k^p, t^{\dagger}\right)$.
    The solution of \eqref{eq:P-STC-bd-mp-3} is 
    \begin{equation}
    t^{\dagger}=\check{t}_k^p+\check \tau (\check{t}_k^p),
    \end{equation}
    where $\check \tau (\check{t}_k^p)$ is given by \eqref{eq:tau_check}. If $t^{\dagger}>\check{t}_k^p+\tau_d$, the next event can be chosen as $\check{t}_{k+1}^p=t^{\dagger}$.
    If $t^{\dagger} \leq \check{t}_k^p+\tau_d$, the next event can be chosen as $\check{t}_{k+1}^p=$ $\check{t}_k^p+\tau_d$. In this way, from \eqref{eq:P-STC-trigger-t}-\eqref{eq:tau_check}, it is ensured that $\Gamma^p(t) \leq 0$ for $t \in\left[\check{t}_k^p, \check{t}_{k+1}^p\right)$, as stated in R2. Applying the similar analysis in the proof of Theorem \ref{thm:P-PETC}, we can obtain $m^p(t)>0$ for all $t>0$ as stated in R3, which implies the positive definiteness of $V^p(t)$ and leads to results stated in R4 and R5.
\end{proof}

\section{Simulation}
\label{sec:sim}

\subsection{Simulation Setup}
We consider the ARZ model with $\gamma=1$, $c_0=0.396$, and the steady-state in congested regime is $(\rho^\star, v^\star)= ($120 vehicles/km, 36 km/h$)$. We use sinusoid initial conditions given by $\rho(x, 0)=0.1 \sin \left({3 \pi x}/{\ell}\right) \rho^{\star}+\rho^{\star}$ and $v(x, 0)=-0.1 \sin \left({3 \pi x}/{\ell}\right) v^{\star}+v^{\star}$. The length of the freeway section is $\ell = 1$ km. The free speed is $v_f = 144$ km/h and the maximum density is $\rho_m = 160$ vehicles/km. The relaxation time is $\tau=2$ minutes. We perform the simulation on a time horizon of 60 minutes.

The parameters for the triggering mechanisms are chosen as follows: The parameters \(\theta\) and \(\sigma\) are set to \(\theta = 1\) and \(\sigma = 0.9\), respectively. The parameters \(\kappa_i\) are calculated from \eqref{betas} as \(\kappa_1 = 280.76\), \(\kappa_2 = 807.29\), and \(\kappa_3 = 2416.5\). The parameter \(\mu\) is set to \(\mu = 11.5\), and the parameter \(C\) is chosen as \(C = 8897.4\) to satisfy \eqref{CC}. The parameter \(\theta_m\), calculated from \eqref{vvbnml}, is \(\theta_m = 1.8705 \times 10^6\). The parameter \(\eta\) is chosen as \(\eta = 1.293\), ensuring that \(\eta = b\), where \(b\) is given by \eqref{bbcfgj}. The initial condition of $m(t)$ is chosen as $m(0)=0.1$. The MDT $\tau_d$ computed from \eqref{eq:tau_d} is \(4.7305 \times 10^{-6}\) hours. Thus, we use \(\Delta t = 4 \times 10^{-6}\) hours to time-discretize the plant and observer dynamics. Following \eqref{eq:h-def}, we also select \(4 \times 10^{-6}\) hours as the sampling period for the PETC approach. Space discretization is performed using a step size of \(\Delta x = 0.005\) km.

\subsection{Comparison of System Behavior}

In Fig. \ref{fig:cmp-all}, we compare R-ETC with P-ETC in terms of the behavior of the Lyapunov function, the control updates, and the dwell times. We set $c=10$ for the P-ETC approaches. It can be observed that the Lyapunov functions for R-ETC are monotonically decreasing. Conversely, the Lyapunov functions for P-ETC approaches sometimes increase, illustrating the flexibility of this approach. Notably, even though the Lyapunov functions under P-ETC approaches converge to zero slower than their regular counterparts, they remain below the performance barrier at all times, thereby meeting the nominal performance. Due to the flexibility of P-ETC Lyapunov functions, control updates under P-CETC and P-PETC are sparser than their regular counterparts. Meanwhile, P-STC slightly outperforms R-STC in terms of update sparsity.

We present the evolution of the system's density, \( \rho(x,t) \), and velocity, \( v(x,t) \), under open-loop, R-CETC, and P-CETC configurations in Fig.~\ref{fig:cmp-rhov}.  
The open-loop system exhibits unstable density-velocity oscillations, as shown in Fig.~\ref{fig:openloop-rho2} and Fig.~\ref{fig:openloop-v2}, indicating that vehicles enter acceleration-deceleration cycles influenced by stop-and-go waves. R-CETC effectively suppresses these oscillations by applying VSL boundary control at the end of the road segment. In contrast, P-CETC sacrifices some suppression of oscillations in favor of sparser triggering, as it deviates from the strict decrease of the Lyapunov function.

Both R-ETC and P-ETC approaches, derived using the linearized ARZ model, are susceptible to Zeno behavior when applied to the nonlinear ARZ model because the nonlinear dynamics were not considered in the exclusion of Zeno behavior. Furthermore, P-ETC approaches (P-CETC, P-PETC, and P-STC) applied to the nonlinear ARZ model is prone to performance barrier violations since nonlinear components are not accounted for in the definition of the performance barrier, \( e^{-b^\star t} V_0^p \), the performance residual, \( W^p(t) := e^{-b^\star t} V_0^p - V^p(t) \), or in the Lyapunov function $V^p(t)$.

\begin{figure}[thbp]
   \centering
    \begin{subfigure}{0.46\textwidth}
        \centering
        {\includegraphics[width=1.1\linewidth]{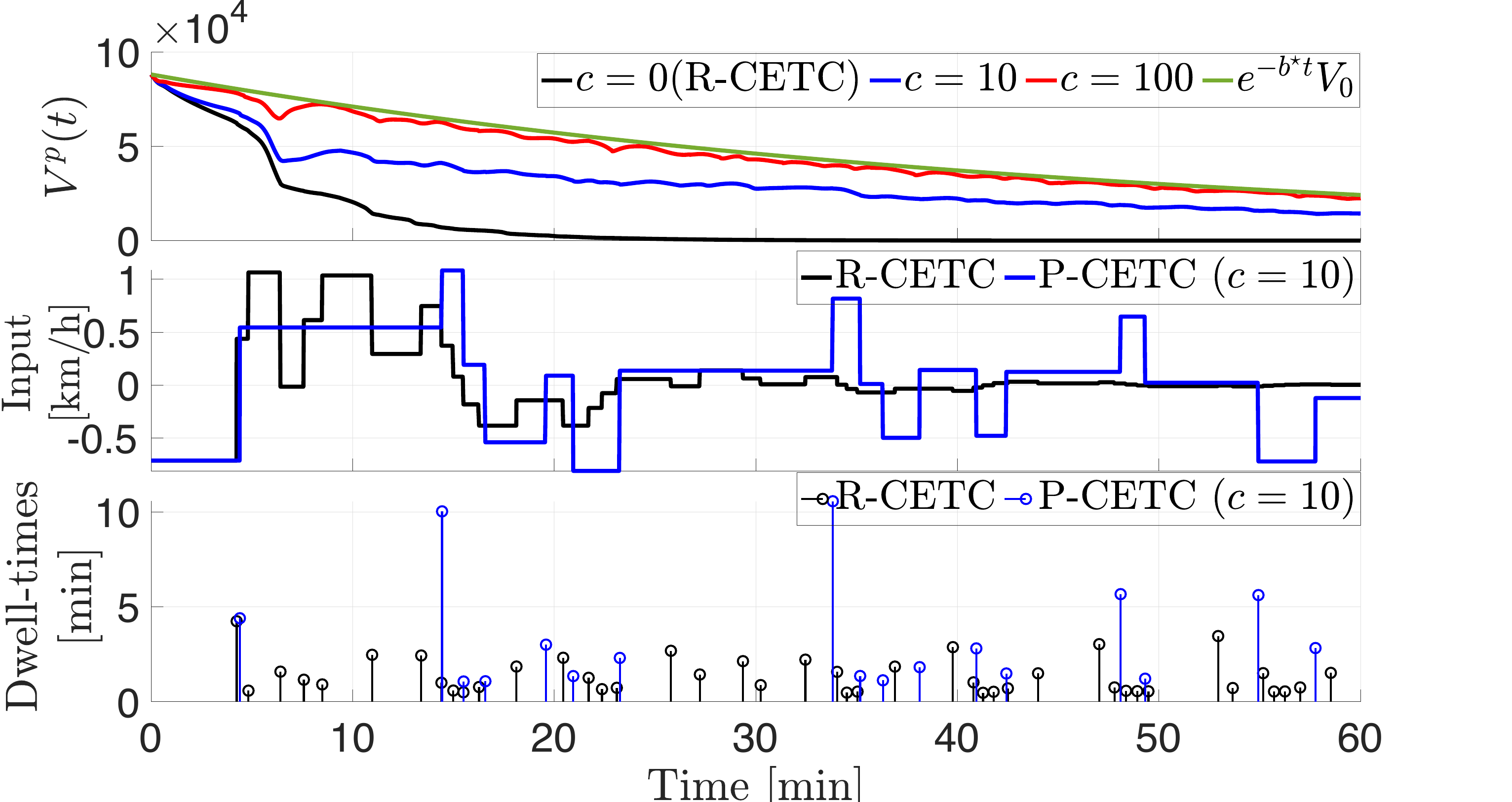}
        }
        \subcaption{Results under R/P-CETC}
        \label{fig:cmp-CETC}
    \end{subfigure}

    \begin{subfigure}{0.46\textwidth}
        \centering
        {\includegraphics[width=1.1\linewidth]{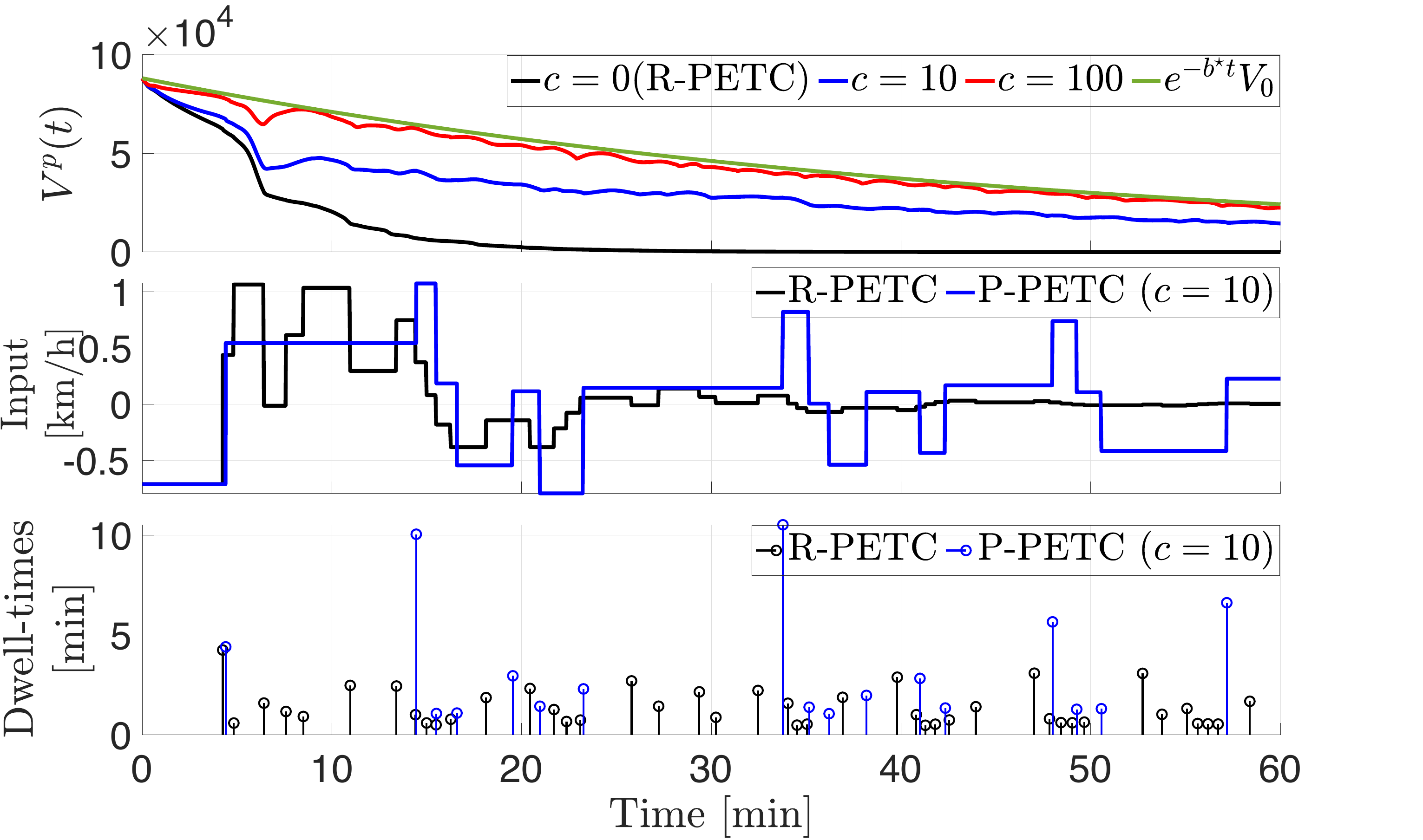}
        }
        \subcaption{Results under R/P-PETC}
        \label{fig:cmp-PETC}
    \end{subfigure}
    
    \begin{subfigure}{0.46\textwidth}
        \centering
        {\includegraphics[width=1.1\linewidth]{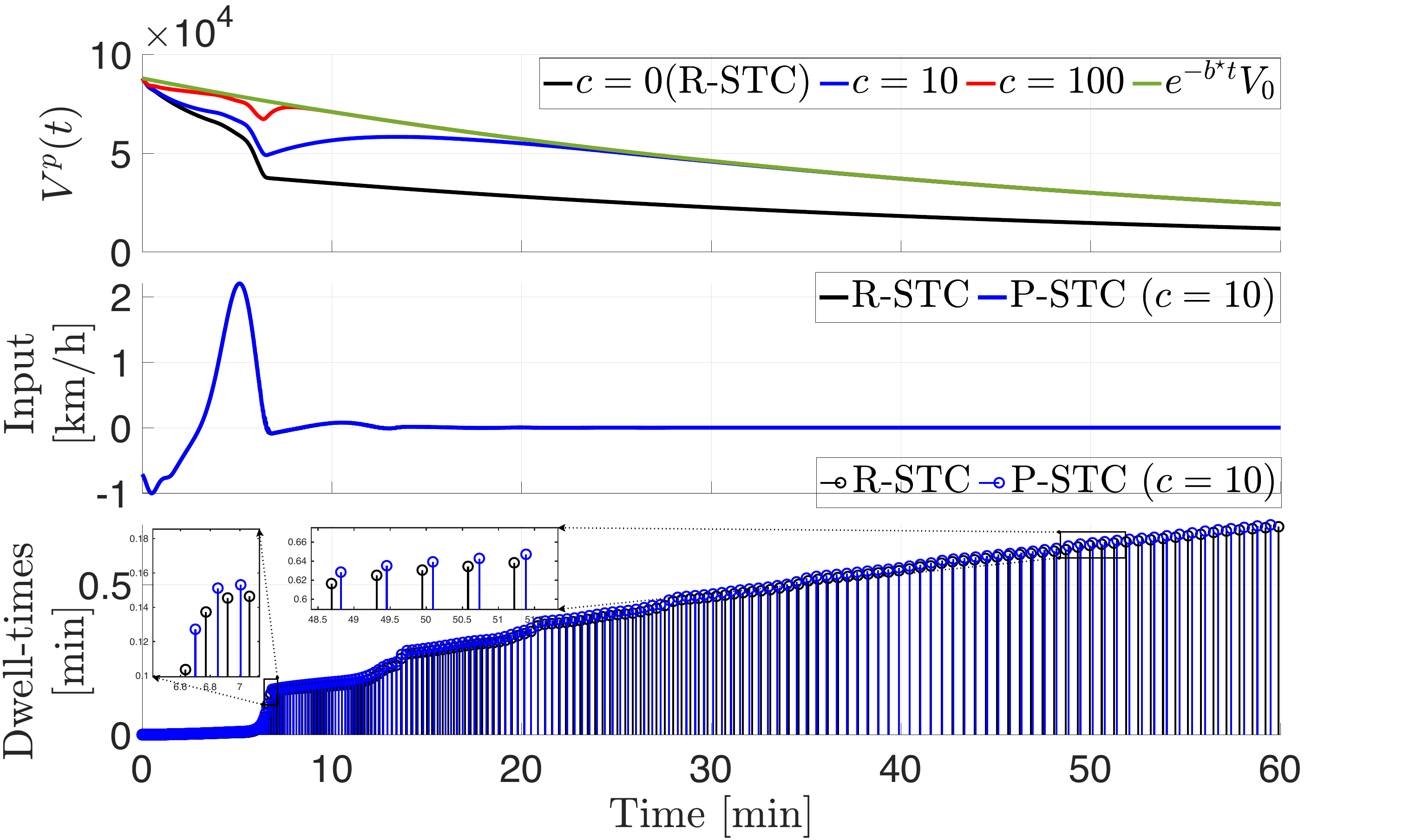}
        }
        \subcaption{Results under R/P-STC}
        \label{fig:cmp-STC}
    \end{subfigure}
    \caption{Comparison of the Lyapunov function, control update and dwell-times under the R/P-CETC, R/P-PETC and  R/P-STC.}
    \label{fig:cmp-all}
\end{figure}

\begin{figure}[htbp]
\centering
\begin{subfigure}{0.24\textwidth}
    \centering
    {\includegraphics[width=1.05\linewidth]{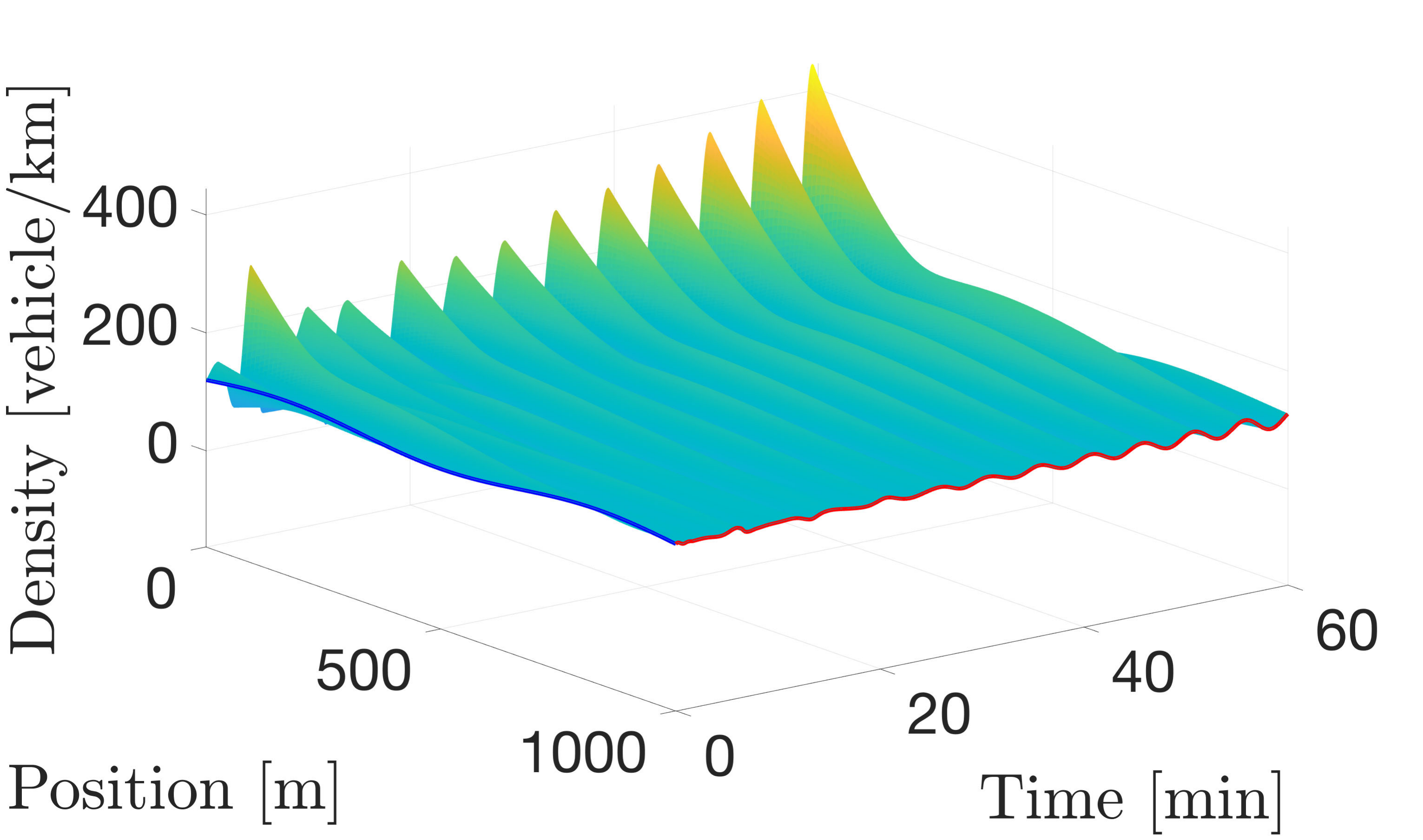}
    }
    \subcaption{$\rho(x,t)$ of open-loop}
    \label{fig:openloop-rho2}
\end{subfigure}
\begin{subfigure}{0.24\textwidth}
    \centering
    {\includegraphics[width=1.05\linewidth]{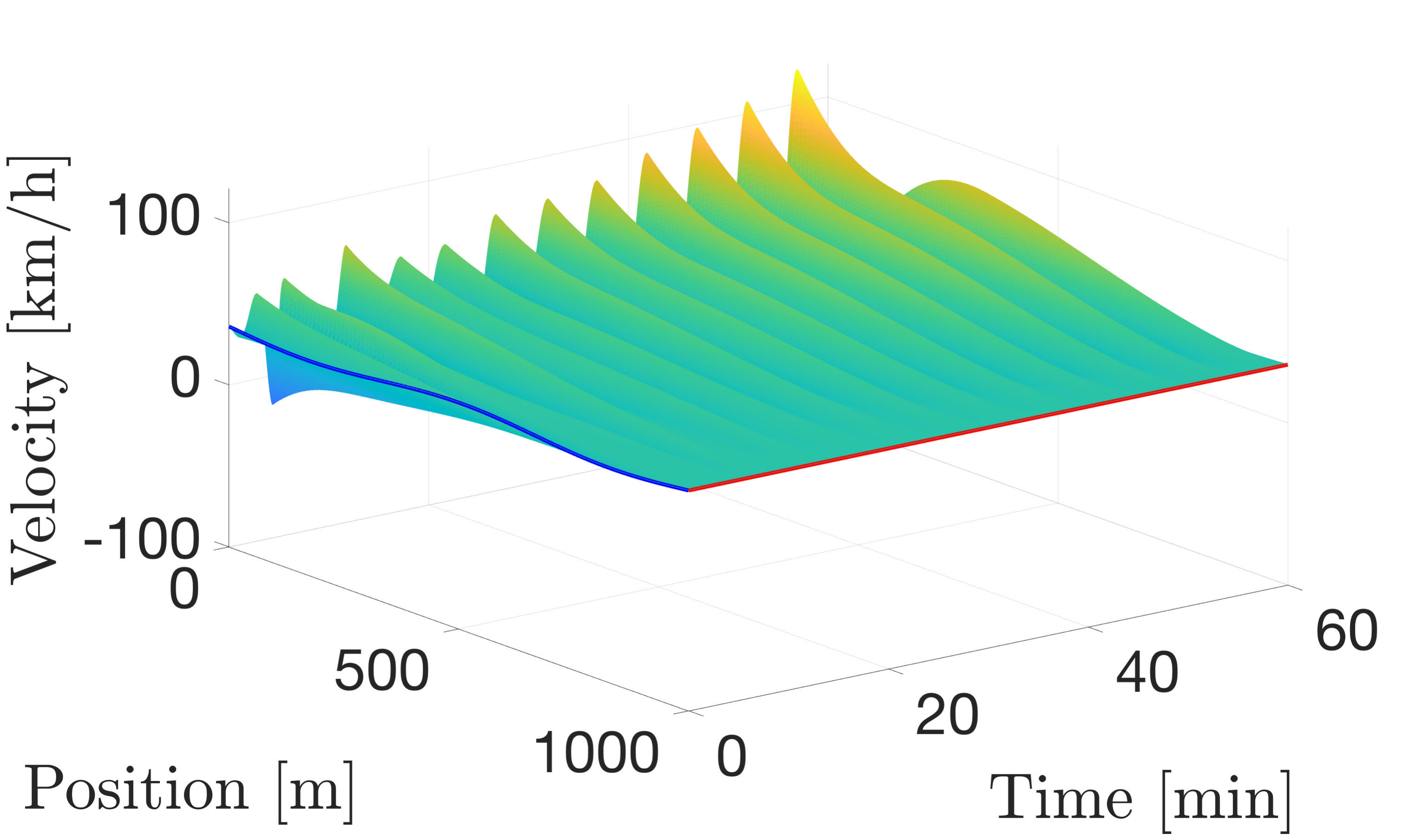}
    }
    \subcaption{$v(x,t)$ of open-loop}
    \label{fig:openloop-v2}
\end{subfigure}

\begin{subfigure}{0.24\textwidth}
    \centering
    {\includegraphics[width=1.05\linewidth]{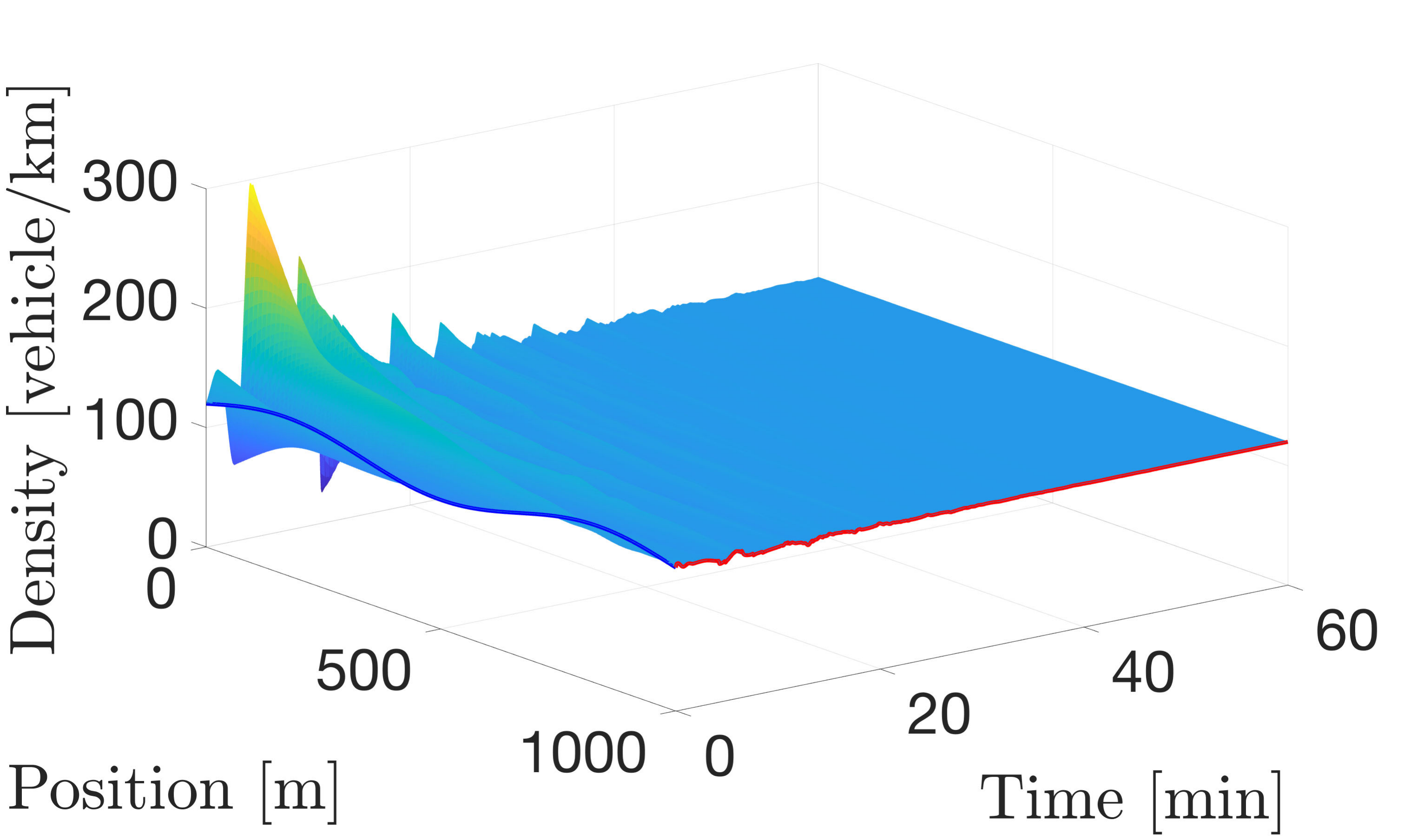}
    }
    \subcaption{$\rho(x,t)$ of R-CETC}
    \label{fig:R-PETC-rho2}
\end{subfigure}
\begin{subfigure}{0.24\textwidth}
    \centering
    {\includegraphics[width=1.05\linewidth]{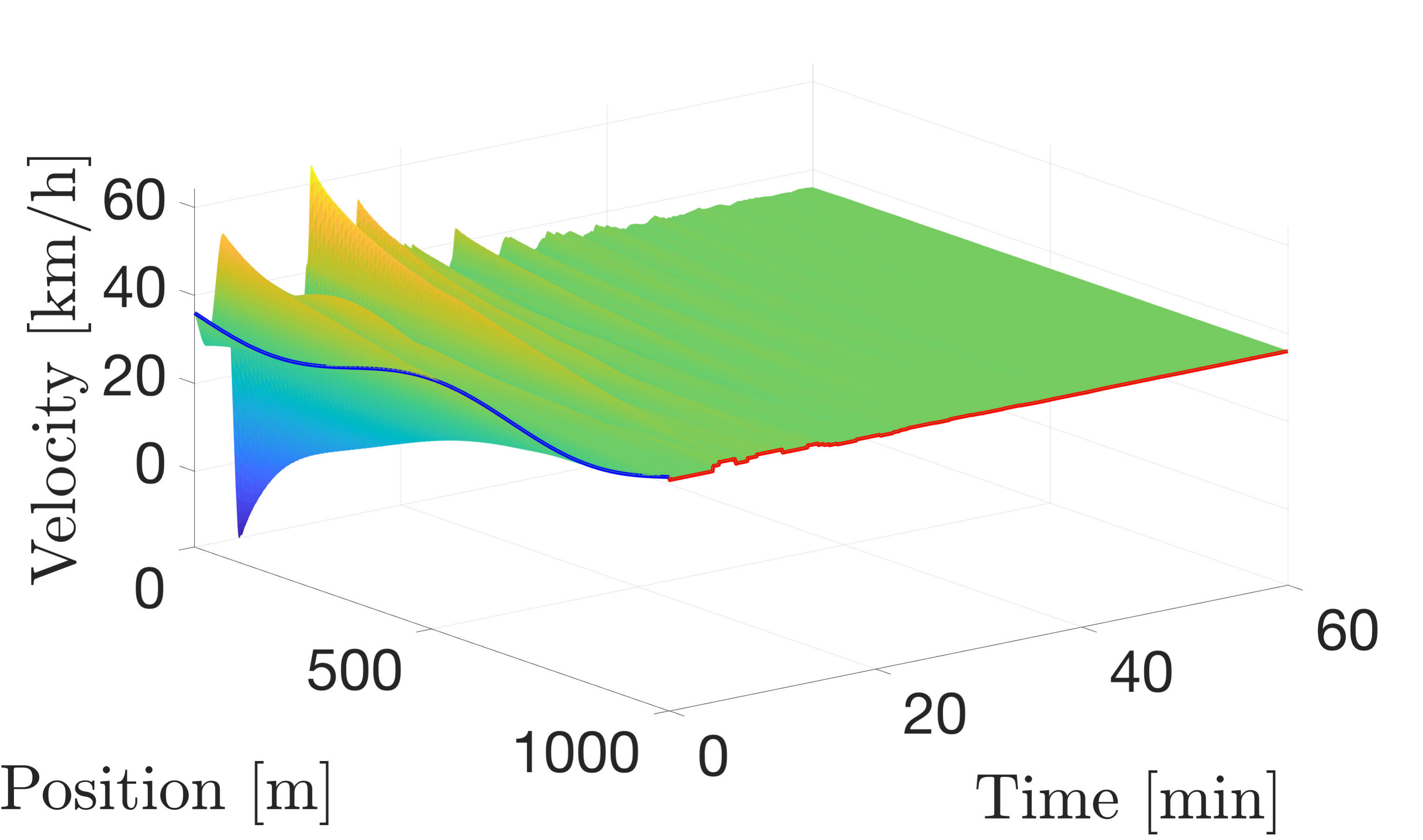}
    }
    \subcaption{$v(x,t)$ of R-CETC}
    \label{fig:R-PETC-v2}
\end{subfigure}

\begin{subfigure}{0.24\textwidth}
    \centering
    {\includegraphics[width=1.05\linewidth]{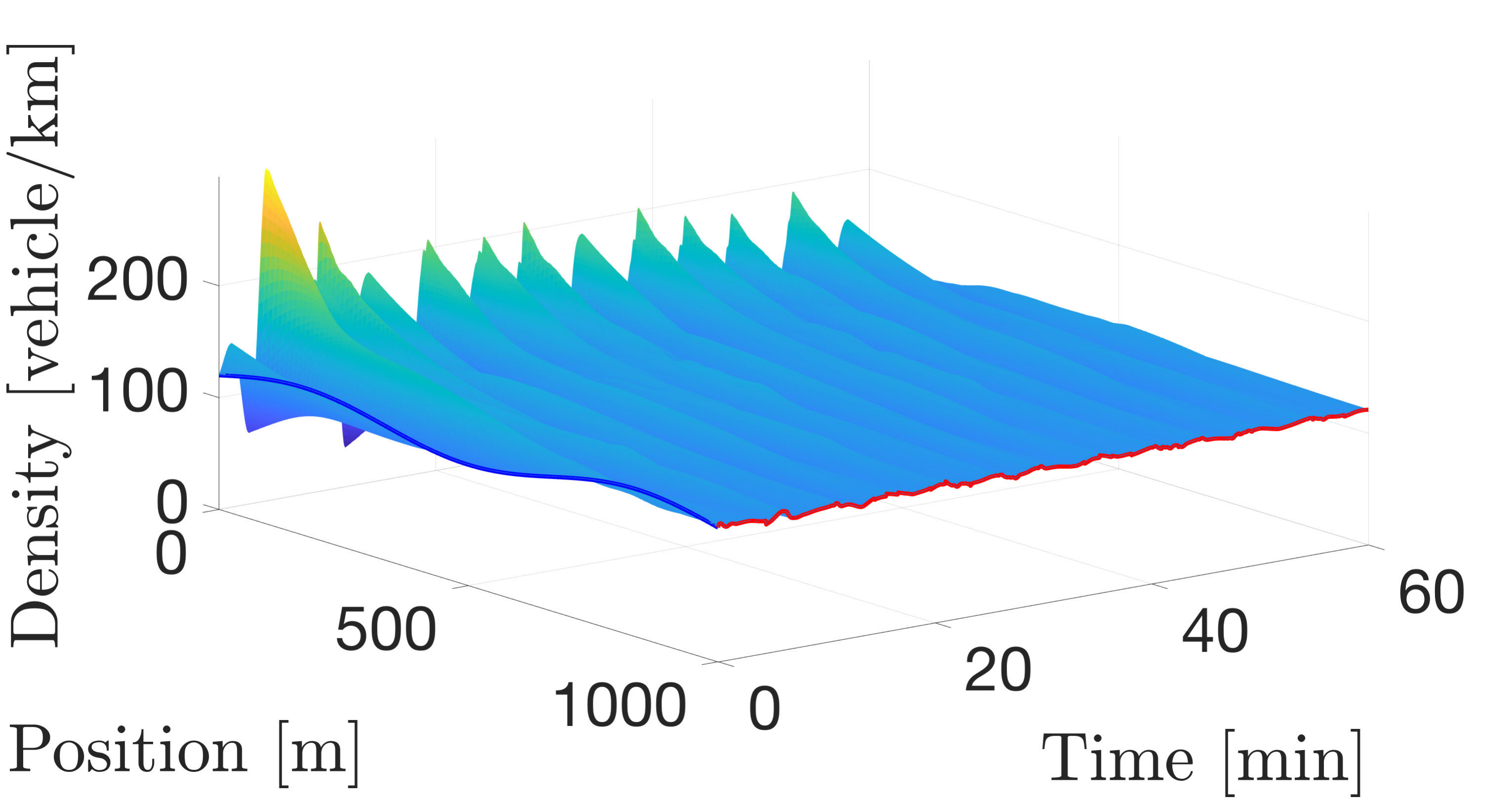}
    }
    \subcaption{$\rho(x,t)$ of P-CETC}
    \label{fig:P-PETC-rho2}
\end{subfigure}
\begin{subfigure}{0.24\textwidth}
    \centering
    {\includegraphics[width=1.05\linewidth]{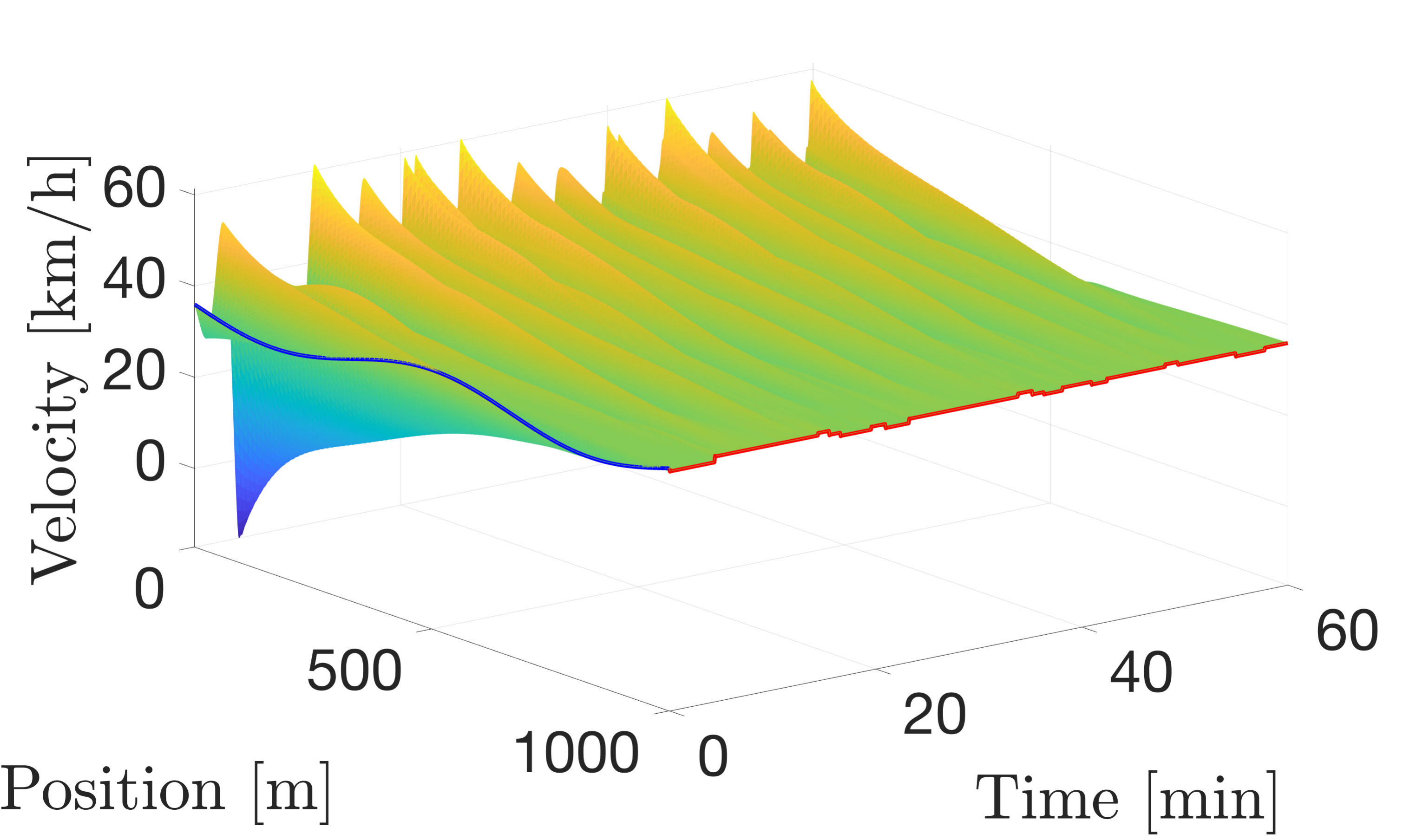}
    }
    \subcaption{$v(x,t)$ of P-CETC}
    \label{fig:P-PETC-v2}
\end{subfigure}

\caption{Comparison of $\rho(x,t)$ and $v(x,t)$.}
\label{fig:cmp-rhov}
\end{figure}

\subsection{Comparison of System Performance}

In TABLE \ref{tab:all-cmp}, we provide a quantitative comparison of the total triggering number $N_t$, the average dwell-time $\bar {\Delta t_k} $ in minutes and three traffic performance metrics of the different approaches at different $c$ values. We adopt traffic performance metrics used in \cite{yu2021reinforcement}, including total travel time $J_{\text{TTT}}$, fuel consumption $J_{\text {fuel }}$ and travel discomfort $J_{\text {D}}$ to evaluate the proposed methods. These performance metrics are given by
\begin{align} 
J_{\text{TTT}}= 
& \int_0^T \int_0^\ell \rho(x, t) d x d t ,\\ 
J_{\text {fuel }}= 
& \int_0^T \int_0^\ell \max \left\{0, b_0+b_1 v(x, t)+b_3 v^3(x, t)\right. \nonumber\\ 
& \left.+b_4 v(x, t) a(x, t)\right\} \rho(x, t) d x d t \label{eq:J-fuel},\\ 
J_{\text {D}}= 
& \int_0^T \int_0^\ell\left(a(x, t)^2+a_t(x, t)^2\right) \rho(x, t) d x d t ,
\end{align}
where $a(x, t)$ is defined as the local acceleration $a(x, t)=$ $v_t(x, t)+v(x, t) v_x(x, t)$ and $b_i$ are constant coefficients chosen as $b_0=25 \cdot 10^{-3}[1 / \mathrm{s}], b_1=24.5 \cdot 10^{-6}[1 / \mathrm{m}], b_3=$ $32.5 \cdot 10^{-9}\left[\mathrm{1s}^3 / \mathrm{m}^2\right], b_4=125 \cdot 10^{-6}\left[\mathrm{1s}^2 / \mathrm{m}^2\right]$.
Note that higher values for the three performance metrics correspond to increased traffic costs, consequently indicating worse traffic performance. For further details on these traffic metrics, see \cite{treiber2013traffic}.

The parameter \( c > 0 \) can be chosen in P-CETC, P-PETC, and P-STC to achieve a smaller total triggering number \( N_t \) compared to their regular counterparts. Generally, \( N_t \) decreases as \( c \) increases because a larger \( c \) provides the Lyapunov function with greater flexibility to deviate from a monotonic decrease, making control updates less likely to be triggered.

The average dwell time, denoted as $\bar{\Delta t_k}$, serves as a \emph{Safety Index (SI).} For small values of $\bar{\Delta t_k}$, drivers exiting the VSL zone are likely to be required to check and adjust their speed more frequently in response to rapidly changing VSL signs. As a consequence, an increased cognitive burden puts safety at risk and is likely to cause the drivers' speed adjustment errors. In light of the SI measure, P-ETC approaches are safer than their regular counterparts. 
In summary, TABLE \ref{tab:all-cmp} demonstrates a correlation between enhanced safety and an increase of the \emph{resource-aware parameter} $c$. Here, the `resource' being saved is (ironically) the risk inflicted upon the safety of the drivers. Remarkably,  the safety index improvement is 2.9$\times$, 2.9$\times$, and 1.3$\times$  for P-CETC, P-PETC, and P-STC, respectively, relative to their regular counterparts.

\begin{table}
\centering
\caption{Comparison of total triggering number $N_t$, average dwell-times $\bar {\Delta t_k} $ in minutes (safety index), and three traffic performance metrics between open-loop and R/P-ETCs within 60 minutes. A negative percentage implies lower traffic cost and better performance compared with open-loop scenario.}
\begin{tblr}{
  width = \linewidth,
  colspec = {Q[120]Q[50]Q[50]Q[60]Q[100]Q[100]Q[100]},
  cells = {c},
  cell{3}{5} = {r},
  cell{3}{6} = {r},
  cell{3}{7} = {r},
  cell{4}{1} = {r=5}{},
  cell{4}{5} = {r},
  cell{4}{6} = {r},
  cell{4}{7} = {r},
  cell{5}{5} = {r},
  cell{5}{6} = {r},
  cell{5}{7} = {r},
  cell{6}{5} = {r},
  cell{6}{6} = {r},
  cell{6}{7} = {r},
  cell{7}{5} = {r},
  cell{7}{6} = {r},
  cell{7}{7} = {r},
  cell{8}{5} = {r},
  cell{8}{6} = {r},
  cell{8}{7} = {r},
  cell{9}{5} = {r},
  cell{9}{6} = {r},
  cell{9}{7} = {r},
  cell{10}{1} = {r=5}{},
  cell{10}{5} = {r},
  cell{10}{6} = {r},
  cell{10}{7} = {r},
  cell{11}{5} = {r},
  cell{11}{6} = {r},
  cell{11}{7} = {r},
  cell{12}{5} = {r},
  cell{12}{6} = {r},
  cell{12}{7} = {r},
  cell{13}{5} = {r},
  cell{13}{6} = {r},
  cell{13}{7} = {r},
  cell{14}{5} = {r},
  cell{14}{6} = {r},
  cell{14}{7} = {r},
  cell{15}{5} = {r},
  cell{15}{6} = {r},
  cell{15}{7} = {r},
  cell{16}{1} = {r=5}{},
  cell{16}{5} = {r},
  cell{16}{6} = {r},
  cell{16}{7} = {r},
  cell{17}{5} = {r},
  cell{17}{6} = {r},
  cell{17}{7} = {r},
  cell{18}{5} = {r},
  cell{18}{6} = {r},
  cell{18}{7} = {r},
  cell{19}{5} = {r},
  cell{19}{6} = {r},
  cell{19}{7} = {r},
  cell{20}{5} = {r},
  cell{20}{6} = {r},
  cell{20}{7} = {r},
  vline{2} = {1-Z}{},
  hline{1-3,9,15,21} = {-}{},
  hline{4,10 ,16} = {-}{dashed},
}
          & $c$  & $N_t$ & $\bar {\Delta t_k}$ & $J_{\text{TTT}}$   & $J_{\text{fuel}}$  & $J_{\text{D}}$     \\
\!\!Open-loop\! & -    & -     & -                   & 4.41$\times10^{5}$ & 1.11$\times10^{4}$ & 4.14$\times10^{5}$ \\
R-CETC    & 0    & 43   & 1.361               & \!\!\!-1.52\%      & \!\!\!-1.40\%      & \!\!\!-80.34\%     \\
P-CETC    & 0.01 & 29    & 1.955               & \!\!\!-1.54\%      & \!\!\! -1.42\%      & \!\!\! -80.26\%     \\
          & 0.1  & 26    & 2.159               & \!\!\! -1.57\%      & \!\!\!-1.44\%      & \!\!\!-79.66\%     \\
          & 1    & 17    & 3.370              & \!\!\!-1.38\%      & \!\!\!-1.26\%      & \!\!\!-74.49\%     \\
          & 10   & 17    & 3.398               & \!\!\! -1.27\%      & \!\!\!-1.16\%      & \!\!\! -56.51\%     \\
          & 100  & 15    & 3.831               & \!\!\!-1.09\%      & \!\!\!-0.98\%      & \!\!\! -45.73\%     \\
R-PETC    & 0    & 43    & 1.358               & \!\!\!-1.52\%      & \!\!\! -1.40\%       & \!\!\!-80.34\%     \\
P-PETC    & 0.01 & 29    &  1.955               & \!\!\! -1.54\%      & \!\!\!-1.42\%      & \!\!\!-80.26\%     \\
          & 0.1  & 26    & 2.159               & \!\!\!-1.56\%      & \!\!\!-1.44\%      & \!\!\!-79.67\%     \\
          & 1    & 17    & 3.370               & \!\!\!-1.38\%      & \!\!\!-1.26\%      & \!\!\!-74.50\%     \\
          & 10   & 17    & 3.363               & \!\!\!-1.26\%      & \!\!\!-1.14\%      & \!\!\!-57.36\%     \\
          & 100  & 15    & 3.832              & \!\!\!-1.08\%      & \!\!\!-0.97\%      & \!\!\! -45.72\%     \\
R-STC     & 0    & 4420   & 0.0137               & \!\!\! -1.74\%      & \!\!\!-1.61\%      & \!\!\!-92.41\%     \\
P-STC     & 0.01 & 4420  & 0.0137              & \!\!\! -1.74\%      & \!\!\!-1.61\%      & \!\!\!-92.41\%     \\
          & 0.1  & 4416   &  0.0136               & \!\!\!-1.74\%      & \!\!\! -1.61\%      & \!\!\!-92.41\%     \\
          & 1    & 4384   &  0.0137               & \!\!\!-1.74\%      & \!\!\!-1.61\%      & \!\!\!-92.41\%     \\
          & 10   &  4148   &  0.0145               & \!\!\! -1.74\%      & \!\!\!-1.61\%      & \!\!\!-92.41\%     \\
          & 100  & 3345   & 0.018              & \!\!\! -1.74\%      & \!\!\! -1.61\%      & \!\!\! -92.41\%     
\end{tblr}
\label{tab:all-cmp}
\end{table}

The comparison of the three performance metrics of the proposed methods is listed in the TABLE \ref{tab:all-cmp}.
The total travel time ($J_{\text{TTT}}$) and fuel consumption ($J_{\text{fuel}}$) for the three P-ETC approaches are reduced by at least 1\% compared to the open-loop system, with the influence of $c$ being negligible, which is consistent with their R-ETC counterparts. Hence, reducing stop-and-go traffic has almost no effect on fuel consumption or travel time, at least for the linear ARZ model. The travel discomfort ($J_{\text{D}}$) is significantly reduced by both the R/P-CETC and R/P-PETC approaches compared to the open-loop system, with reductions ranging between 45\% and 80\%, as the controls suppress stop-and-go oscillations in the closed-loop system. The largest decrease in travel discomfort ($J_{\text{D}}$) is observed with the R/P-STC approach, which achieves a reduction of around 92\% compared to the open-loop system. This is because R/P-STC results in frequent control updates, closely emulating continuous-time control.

Even though the travel discomfort somewhat increases with larger $c$ for P-CETC and P-PETC compared with their regular counterparts, they still achieve a considerable discomfort reduction compared with the open-loop control. Therefore, when selecting the parameter $c$ for P-ETC, a balance should be sought between system performance metrics. The increasing $c$ has no adverse effect on the traffic metrics when contrasting R-STC with P-STC. In the cases of P-CETC and P-PETC, they may accomplish fewer $N_t$ and better safety at the expense of somewhat increased but still satisfactory travel discomfort.

In conclusion, the proposed P-ETCs reduce travel discomfort by 45\%-92\% relative to driver's natural behavior (open-loop) and increase driver safety, measured by the average dwell time, by as much as 2.9$\times$ relative to their regular counterparts with the frequent-switching VSL schedule.

\section{Conclusions}
\label{sec:Conclusions}
This paper has employed the recently introduced ETC approach known as performance-barrier ETC (P-ETC) to control the linearized Aw-Rascle-Zhang traffic model, a $2\times 2$ coupled hyperbolic PDE equipped with a varying speed limit. We have explored P-ETC across three configurations: continuous-time ETC (P-CETC), periodic ETC (P-PETC), and self-triggered control (P-STC). Unlike the existing regular ETC (R-ETC), where the closed-loop system's Lyapunov function is forced to decrease constantly, the proposed P-ETC allows the Lyapunov function of the closed-loop system to deviate from strict monotonic decrease, provided it remains below an acceptable performance barrier. This flexibility results in extended dwell times between events compared with R-ETC.
We have also presented PETC and STC variants of R-ETC, which have not been previously explored for coupled hyperbolic PDEs. We have demonstrated that all proposed control approaches guarantee exponential convergence to zero in the spatial $L^2$ norm while ensuring Zeno-free behavior. The performance of the proposed methods has been illustrated through numerical simulations, and extensive comparisons between different methods have been provided, focusing on triggering number, driver's safety and traffic metrics such as vehicle fuel consumption, total travel time, and driver comfort.

For future work, we aim to investigate event-triggered control under quantization effects for traffic phenomena, which may enable more practical advisory speeds with variable speed limits (VSLs). Furthermore, event-triggered control for nonlinear ARZ PDEs is of interest, as nonlinear effects in traffic flow phenomena were not considered in this work.

\section*{Appendix: Proof of Theorem \ref{thm:R-CETC} (R-CETC)}

To streamline the proof of Theorem \ref{thm:R-CETC}, we first present Lemmas \ref{lem1} and \ref{mdt_lem}.  

\begin{lem}\label{lem1}Under the R-CETC approach \eqref{eq:U-R-CETC}-\eqref{eq:ODE-dot-m-r}, it holds that $\Gamma^r(t):=d^2(t)-\theta m^r(t)\leq 0$ and $m^r(t)> 0,$ for all $t\in [0,\sup(I^r))$, where $I^r=$ $\left\{t_0^r, t_1^r, t_2^r, \ldots\right\}$.
\end{lem}

The proof is similar to that of Lemma 1 of \cite{espitiaEventBasedBoundaryControl2018}, and is hence omitted.

\begin{lem}\label{mdt_lem}Under the R-CETC approach \eqref{eq:U-R-CETC}-\eqref{eq:ODE-dot-m-r}, with $\kappa_{1},\kappa_{2},\kappa_{3}>0$ chosen as in \eqref{betas}-\eqref{gg3},  there exists a uniform minimal dwell-time $\tau_d>0$, given by \eqref{eq:tau_d}-\eqref{eq:eps-2}, between two triggering times, \textit{i.e.,} there exists a constant $\tau_d>0$ such that $t_{k+1}^r-t_{k}^r\geq\tau_d,$ for all $j\in\mathbb{N}$.\end{lem}

The proof is similar to that of Theorem 1 of \cite{espitiaEventBasedBoundaryControl2018}, and is hence omitted. 

Due to the existence of a minimal dwell time $\tau_d$ guaranteed by Lemma \ref{mdt_lem}, Zeno behavior is absent, thereby proving R1. Next, R2 follows directly from Proposition 1 and Remark 1. Furthermore, the existence of solutions for all $t > 0$ implies that R3 follows from Lemma \ref{lem1}.

Now, let us proceed with the proofs of R4 and R5. Taking the time derivative of $V_1(t)$ given by \eqref{eq:Lyap-V} for all $t\in(t_k^r,t_{k+1}^r),j\in\mathbb{N}$ and integrating by parts, we obtain that
\begin{align}\label{dfg1dfnmhj}
\begin{split}
  \dot{V}_1(t) = & -\mu V_1(t)-C\alpha^2(\ell,t)e^{-\frac{\mu\ell}{v^{\star}}}+Cr_0^2r_1^2d^2(t)e^{\frac{\mu\ell}{(\gamma p^{\star}-v^{\star})}}.
\end{split}
\end{align}
Then, considering \eqref{eq:Lyap-V-r} and \eqref{eq:ODE-dot-m-r}, we can write
\begin{align}\label{fgshlhj}
\begin{split}
    \dot{V}^r(t) &= \dot{V}_1(t)+\dot{m}^r(t)\\
    &=-\mu V_1(t)-C\alpha^2(\ell,t)e^{-\frac{\mu\ell}{v^{\star}}}+Cr_0^2r_1^2d^2(t)e^{\frac{\mu\ell}{(\gamma p^{\star}-v^{\star})}}\\&\quad-\eta m^r(t)- \theta_m d^2(t)+\kappa_1\Vert\alpha[t]\Vert^2+\kappa_2\Vert\beta[t]\Vert^2\\&\quad+\kappa_3{\alpha}^2(\ell, t),
\end{split}
\end{align}
for $t\in(t_k^r,t_{k+1}^r),j\in\mathbb{N}$. Note from \eqref{eq:Lyap-V} that
\begin{align}\label{xxnmkihj}
    \Vert \alpha[t]\Vert^2+\Vert\beta[t]\Vert^2 \leq \frac{r}{C}V_1(t),
\end{align}
for all $t\geq 0$, where $r$ is given by \eqref{rr}. Then, we can obtain from \eqref{fgshlhj} that
\begin{align}
    \begin{split}
        \dot{V}^r(t) \leq&\hspace{-5pt}-\Big(\mu-\hspace{-3pt}\frac{\max\{\kappa_1,\kappa_2\}r}{C}\Big)V_1(t)-\hspace{-5pt}\Big(Ce^{-\frac{\mu\ell}{v^\star}}-\kappa_3\Big)\alpha^2(\ell,t)\\&-\Big(\theta_m-Cr_0^2r_1^2e^{\frac{\mu\ell}{(\gamma p^{\star}-v^{\star})}}\Big)d^2(t)-\eta m^r(t),
    \end{split}
\end{align}
for $t\in(t_k^r,t_{k+1}^r),j\in\mathbb{N}$. Recalling that $C>0$ is chosen such that \eqref{CC} is satisfied and $\theta_m>0$ is chosen as in \eqref{vvbnml}, we can obtain that
\begin{equation}
    \dot{V}^r(t) \leq -b^\star V^r(t),
\end{equation}
for $t\in(t_k^r,t_{k+1}^r),j\in\mathbb{N}$, where $b^\star>0$ is given by \eqref{eq:b*},\eqref{bbcfgj}. Then, considering the time continuity of $V^r(t)$, we can obtain \eqref{eq:Lyap-V-r-bd} valid for all $t>0$. Further, by following classical arguments involving the bounded invertibility of the backstepping transformations \eqref{eq:K1},\eqref{eq:K2},\eqref{eq:L1},\eqref{eq:L2}, we obtain R5.

\bibliographystyle{IEEEtranS}
\bibliography{main.bib}

\end{document}